\documentclass[1p]{elsarticle}

\usepackage{amssymb}
\usepackage{amsmath}
\usepackage[procnumbered, linesnumbered,algoruled]{algorithm2e}

\usepackage{color}
\usepackage{tikz}
\usetikzlibrary{patterns,decorations.pathreplacing,decorations.markings,math,arrows.meta}
\usepackage{floatrow}

 \newtheorem{theorem}{Theorem} \newtheorem{lemma}[theorem]{Lemma}
 \newtheorem{claim}[theorem]{Claim}
  \newtheorem{definition}[theorem]{Definition} \newdefinition{rmk}{Remark} \newproof{proof}{Proof}

\newcommand{\cC}{{\mathcal{C}}}
\newcommand{\bx}{{\bar{x}}}

\newcommand{\bz}{{\bar{z}}}

\newcommand{\OPT}{\textnormal{OPT}}

\newcommand{\mS}{\mathcal{S}}
\newcommand{\mL}{\mathcal{L}}
\newcommand{\eps}{{\varepsilon}}
\newcommand{\aeps}{{\mu}}
\newcommand{\E}{{\mathbb{E}}}

\DeclareMathOperator*{\argmin}{arg\,min}
\def \II   {{\mathcal I}}
\newcommand{\floor}[1]{{\left\lfloor {#1} \right\rfloor}}

\begin{document}
\title{An Almost Optimal
	Approximation Algorithm for Monotone Submodular Multiple Knapsack
	\tnoteref{prel_version}} 
\tnotetext[prel_version]{
A preliminary version of
this paper appeared in Proc. of the 28th Annual European Symposium on Algorithms (ESA), Pisa (Virtual Conference), September 2020.}


\author[1]{Yaron Fairstein}\ead{yyfairstein@gmail.com}
\author[1]{Ariel Kulik\fnref{fn1}}\ead{kulik@cs.technion.ac.il}
\author[1]{Joseph (Seffi) Naor\fnref{fn2}}\ead{naor@cs.technion.ac.il}

\author[1]{Danny Raz}\ead{danny@cs.technion.ac.il}
\author[1]{Hadas Shachnai}\ead{hadas@cs.technion.ac.il}

\affiliation[1]{organization={Computer Science Department, Technion}, 
	postcode={3200003}, city={Haifa}, country={Israel}}

\fntext[fn2]{This research was supported in part by US-Israel BSF grant 2018352 and by ISF grant 2233/19 (2027511)}
\fntext[fn1]{Corresponding author.}
\begin{abstract}
We study the problem of maximizing a monotone submodular function subject to a Multiple Knapsack constraint. The input is a set $I$ of items, each has a non-negative weight, and a set of bins of arbitrary capacities.
Also, we are given a submodular, monotone and non-negative function $f$ over subsets of the items. The objective is to find a packing of a subset of items  $A \subseteq I$ in the bins such that $f(A)$ is maximized.

Our main result is an almost optimal polynomial time $(1-e^{-1}-\eps)$-approximation algorithm for the problem, for any $\eps>0$. The algorithm relies on 
a structuring technique which converts a general multiple knapsack constraint to a constraint in which the bins are  partitioned into groups  of exponentially increasing cardinalities, each consisting of bins of uniform capacity. We derive the result by combining structuring with  a refined analysis of techniques for submodular optimization subject to knapsack constraints. 
\end{abstract}
\maketitle


\section{Introduction}

Submodular optimization has recently attracted much attention as it provides a unifying framework
capturing many fundamental problems in combinatorial optimization, economics, algorithmic game theory, networking, and
other areas. Furthermore, submodularity also captures many real-world practical applications where economy of scale is prevalent.
Classic examples of submodular functions are coverage functions \cite{feige1998threshold}, matroid rank functions \cite{calinescu2007maximizing} and graph cut functions \cite{feige1995approximating}. A recent survey on submodular functions can be found in \cite{buchbinder2017submodular}.

Submodular functions are defined over sets. Given a ground set 
$I$, a function
$f:2^{I} \rightarrow \mathbb{R}_{\geq 0}$
is called {\em submodular} if
for every $A\subseteq B\subseteq I$ and $i\in I \setminus B$, $f(A \cup\{i\}) - f(A) \geq f(B \cup\{i\}) - f(B)$.\footnote{Equivalently, for every $A,B\subseteq I$: $f(A)+f(B)\geq
	f(A\cup B) + f(A\cap B)$.}
This reflects the diminishing returns property: the marginal value from adding $i\in I$ to a solution diminishes as the solution set becomes larger.
A set function $f:2^I\rightarrow \mathbb{R}$ is {\em monotone} if for any $A\subseteq B \subseteq I$ it holds that $f(A)\leq f(B)$. While in many cases, such as coverage and  matroid rank function, the submodular function is monotone, this is not always the case (cut functions are a classic example).

The focus of this work is  optimization of monotone submodular functions.
In \cite{nemhauser1978best} Nemhauser  and  Wolsey
presented a greedy based $(1-e^{-1})$-approximation for  maximizing a monotone submodular function subject to a cardinality constraint, along with a matching lower bound in the oracle model.
A $(1-e^{-1})$ hardness of approximation bound is also known for the problem under $\textsc{P}\neq \textsc{NP}$, due to the hardness of max-$k$-cover \cite{feige1998threshold} which is a special case. The greedy algorithm of \cite{nemhauser1978best} was later generalized to monotone submodular optimization with a knapsack constraint \cite{khuller1999budgeted,sviridenko2004note}.

A major breakthrough in the field was the continuous greedy algorithm presented in \cite{vondrak2008optimal}. Initially used to derive a $(1-e^{-1})$-approximation for maximizing a monotone submodular function subject to a matroid constraint, the algorithm has become a primary tool in the development of monotone submodular maximization algorithms subject to various other constraints. These include $d$-dimensional knapsack constraints \cite{kulik2013approximations}, and combinations of $d$-dimensional knapsack and matroid constraints \cite{chekuri2010dependent}. A variant of the continuous greedy algorithm for non-monotone functions is given in \cite{feldman2011unified}.

In the {\em multiple knapsack problem} (MKP) we are given a set of items, where each item has a weight and a profit, and a set of bins of arbitrary capacities. The objective is to find a packing of a subset of the items that respects the bin capacities and yields a maximum profit. The problem is one of the most natural extensions of the classic Knapsack problem arising also in the context of Virtual Machine (VM) allocation in cloud computing.  The practical task is to assign VMs to physical machines such that capacity constraints are satisfied, while maximizing the profit of the cloud provider. A submodular cost function allows cloud providers to offer complex cost models to high-volume customers, where the price customers pay for each VM can depend on the overall number of machines used by the customer.

A {\em polynomial time approximation scheme} for MKP was first presented by Chekuri and Khanna~\cite{chekuri2005polynomial}. The authors also ruled out the existence of a {\em fully polynomial time approximation scheme} for the problem. An
{\em efficient} polynomial time approximation scheme was later developed by Jansen~\cite{jansen2010parameterized,jansen2012fast}.

\subsection{Our Results}

In this paper we consider the {\em submodular multiple knapsack problem} (SMKP). The input consists of a set of  $n$ items $I$ and
$m$ bins $B$. Each item $i\in I$ is associated with a weight $w_i\geq 0$, and each bin $b \in B$ has a capacity $W_b\geq 0$. We are also given an oracle to a non-negative monotone submodular function $f:2^I \rightarrow \mathbb{R}_{\geq 0}$.
A feasible solution to the problem is a tuple of $m$ subsets $(A_b)_{b\in B}$ such that for every $b\in B$ it holds that $\sum_{i\in A_b} w_i \leq W_b$. The value of a solution $(A_b)_{b\in B} $ is
$f\left(\bigcup_{b\in B} A_b\right)$. The goal is to find a feasible solution of maximum value.\footnote{We note that the set of bins $B$ is part of the input for SMKP, thus the number of bins is non-constant. This is one  difference between SMKP and the problem of maximizing a submodular set function subject
to $d$ knapsack constraints (or, a $d$-dimensional knapsack constraint) where $d$ is fixed (for more details see, e.g.,~\cite{kulik2013approximations}).}

The problem is a natural generalization of both Multiple Knapsack~\cite{chekuri2005polynomial} (where $f$ is modular or linear), and the problem of monotone submodular maximization subject to a knapsack constraint \cite{sviridenko2004note} (where $m=1$).
Our main result is stated in the next theorem.
\begin{theorem}\label{thm:main_result}
	For any $\eps>0$, there is a randomized $(1-e^{-1}-\eps)$-approximation algorithm for SMKP.
\end{theorem}

As mentioned above, a $(1-e^{-1})$ hardness of approximation bound is known for the problem under $\textsc{P}\neq \textsc{NP}$, due to the hardness of max-$k$-cover \cite{feige1998threshold} which is a special case of SMKP.
This is a vast improvement over previous results.
Feldman presented in~\cite{feldman2013maximization}  a $\left( \frac{e-1}{3e-1} -o(1)\right)\approx 0.24$-approximation for the special case of identical bin capacities, along with a $\frac{1}{9}$-approximation for general capacities.
To the best of our knowledge, this is the best known approximation ratio for the problem.\footnote{Sun et. al. \cite{sun2020tight} indicate that a $\left(1-e^{1-e^{-1}} -o(1)\right)\approx 0.468$-approximation for the problem can be derived using the techniques of \cite{calinescu2011maximizing}. We note that this derivation is non-trivial (no details were given in~\cite{calinescu2011maximizing}).}

 Simultaneously and independently to our work, Sun et. al. \cite{sun2020tight} presented a {\em deterministic} greedy based $(1-e^{-1}-\eps)$-approximation for the special case of identical bins.
In a later version~\cite{sun2020tight_v3},
which appeared after the publication of the preliminary version of this paper, Sun et. al. derived a randomized $(1-e^{-1}-\eps)$-approximation for general SMKP instances,  
matching our result,
by using a different approach.

\subsection{Tools and Techniques} \label{sec:tools_techniques}

Our algorithm relies on a refined analysis of techniques for submodular optimization subject to
$d$-dimensional knapsack constraints~\cite{kulik2013approximations, calinescu2011maximizing,chekuri2010dependent},
combined with sophisticated application of tools
used in the development of approximation schemes for packing problems~\cite{de1981bin}.

At the heart of our algorithm lies the observation that SMKP for a large number of identical bins (i.e., $\forall b\in B,~W_b=W$ for some $W\geq0$) can be easily approximated via a reduction
to the problem of maximizing a submodular function subject to a
$2$-dimensional knapsack constraint (see, e.g., \cite{kulik2013approximations}). Given such an SMKP instance and $\eps>0$, we partition the items to {\em small} and {\em large}, where an item $i\in I$ is small if $w_i \leq \eps W$ and large otherwise. We further define a {\em configuration} to be a subset of large items which fits into a single bin, and let $\cC$ be the set of all configurations. It follows that for fixed $\eps>0$, the number of configurations is polynomial.

Using the above we define
a new submodular optimization problem, to which we refer as the {\em block-constraint} problem. We define a new universe $E$ which consists of all configurations $\cC$ and all small items, $E=\cC\cup \{\{i\}|\textnormal{~$i$ is small}\}$. We also define a new submodular function $g:2^E\rightarrow \mathbb{R}_{\geq 0}$ by
$g(T) = f\left( \bigcup_{A\in T} A\right)$. Now, we seek a subset of elements $T\subseteq E$ such that $T$ has at most $m=|B|$ configurations, i.e., $|T\cap \cC| \leq m$,
and the total weight of sets selected is at most $m\cdot W$; namely, $\sum_{A\in T} w(A)\leq m \cdot W$, where $w(A) = \sum_{i\in A} w_i$.

It is easy to see that the optimal value of the block-constraint problem is at least the value of the optimum for the original instance. Moreover, a solution $T$ for the block-constraint problem can be used to generate a solution for the SMKP instance with only a small loss in value. As there are no more than $m$ configurations, and all other items are small, the items in $T$ can be easily packed into $(1+\eps)m+1$ bins of capacity $W$ using First Fit. Then, it is possible to remove  $\eps m +1$ of the bins while maintaining at least $\frac{m}{	(1+\eps) m +1} \geq \frac{1}{1+2\eps}$ of the solution value, for $m\geq \frac{1}{\eps}$. Once these $\eps m +1$ bins are removed, we have a feasible solution for the SMKP instance.
The block-constraint problem can be viewed as
monotone submodular optimization subject to a $2$-dimensional knapsack constraint. Thus, a  $(1-e^{-1}-\eps)$-approximate solution can be found efficiently \cite{kulik2013approximations}.

Our approximation algorithm for SMKP is based on a generalization of the above. We refer to a set of bins of identical capacity as a {\em block}, and show how to reduce an SMKP instance into a  submodular optimization problem with a $d$-dimensional knapsack constraint,  in which $d$ is twice the number of blocks plus a constant. While, generally, this problem cannot be solved for non-constant $d$, we use a refined analysis of known algorithms~\cite{kulik2013approximations,chekuri2010dependent} to show that the problem can be efficiently solved if the blocks admit a certain structure, to which we refer as {\em leveled}.

We utilize  a grouping technique, inspired by the work of  Fernandez de la Vega and Lueker \cite{de1981bin}, to convert a general SMKP instance to a leveled instance.
We sort the bins in decreasing order by capacity and then partition them into levels, where level $t$, $t\geq 0$, has $N^{2+t}$ bins, divided into $N^2$ consecutive blocks, each containing $N^t$ bins. We decrease the capacity of each bin to the smallest capacity of a bin in the same block. While the decrease in capacity generates the leveled structure required for our algorithm to work, it only slightly decreases the optimal solution value. The main idea is that given an optimal solution, each block of decreased capacity can now be used to store the items assigned to the subsequent block  on the same level. Also, the items assigned to $N$ blocks from each level can be evicted, while only causing a reduction of $\frac{1}{N}$ to the profit (as only $N$ of the $N^2$ blocks of the level are evicted). These evicted blocks are then used for the items assigned to the first block in the next level.

\section{Preliminaries}

\label{sec:prelim}

Our analysis utilizes several basic properties of submodular functions.
Given a  monotone submodular function $f:2^I \rightarrow \mathbb{R}_{\geq 0}$ and  a set $S\subseteq I$, we define $f_S:2^I \rightarrow \mathbb{R}_{\geq 0}$ by $f_S(A) = f(S\cup A )- f(S)$. It follows that $f_S$ is a monotone, non-negative submodular function (see Claim \ref{submodular-residual} in \ref{app:submodular}). The proof of the next claim is given in \ref{app:submodular}.

\begin{claim}\label{submodular-extension}
	Let $f:2^I\rightarrow\mathbb{R}_{\geq 0}$ be a non-negative, monotone and submodular function, and let
	$E \subseteq 2^I \times X$
	for some set $X$ (each element of $E$ is a pair $(S,h)$ with $S\subseteq I$ and $h\in X$). Then the function $g:2^{E}\rightarrow\mathbb{R}_{\geq 0}$ defined by $g(A) =f\left(\cup_{(S,h)\in A} S\right)$ is non-negative, monotone and submodular.
\end{claim}

While Claim~\ref{submodular-extension} is essential for our algorithm, it is important to emphasize it does not hold for non-monotone submodular functions.

Many modern submodular optimization algorithms rely on the submodular {\em Multilinear Extension} (see, e.g., \cite{calinescu2007maximizing,kulik2013approximations,lee2010maximizing,vondrak2013symmetry,feldman2011unified,buchbinder2014submodular}). Given a function $f:2^I\rightarrow\mathbb{R}_{\geq 0}$, its multilinear extension is $F:[0,1]^I\rightarrow\mathbb{R}_{\geq 0}$ defined as:
$$ F(\bx) = \sum_{S\subseteq I}f(S)\prod_{i\in S}\bx_i\prod_{i\in I\setminus S}(1-\bx_i).$$

The multilinear extension can be interpreted
as an expectation of a random variable. Given $\bx\in [0,1]^I$ we say that a random set $X$ is distributed according to $\bx$, $X\sim\bx$, if $\Pr(i\in X)=\bx_i$ and the events $(i\in X)_{i\in I}$ are independent. It follows that $F(\bx)=\E_{X\sim \bx}[f(X)]$.

The unified greedy algorithm of~\cite{feldman2011unified} can be used to find approximate solution for maximization problems of the form
  $\max F(\bx) ~~\text{s.t.}~~ \bx\in P$, where $F$ is the multilinear extension of a monotone submodular function $f$, and $P$ is a down-monotone polytope. The algorithm uses two oracles, one for $f$ and another which given $\bar{\lambda}\in \mathbb{R}^I$ returns
a vector $\bar{x}\in P$ such that $\bar{x}\cdot \bar{\lambda}$ is maximal.
The algorithm returns
$\bx\in P$ such that $F(\bx)\geq \left(1-e^{-1}-o(1)\right) \max_{\bar{y} \in P\cap \{0,1\}^I} F(\bar{y})$. The result can also obtained via the continuous greedy of \cite{calinescu2011maximizing}.

We use $\II=(I,w,B,W,f)$ to denote an SMKP instance consisting of a set of items $I$ with weights $w_i$ for $i\in I$, a set of bins $B$ with capacities $W_b$ for $b\in B$, and objective function $f$. Given a set $A\subseteq I$, let $w(A)= \sum_{i\in A} w_i$. We denote by $\OPT(\II)$ the optimal solution value for the instance $\II$.

\section{The Approximation Algorithm}

In this section we present our approximation algorithm for SMKP. Given an instance ${\II}$ of the problem, let
$A^* = \cup_{b \in B} A^*_b$ be an optimal solution of value $OPT({\II})$. We first observe that there exists a constant size subset $A = \cup_{b \in B} A_b$, where  $A_b \subseteq A^*_b$, satisfying the following property:
the value gained from any item in $i \in A^* \setminus A$ is small relative to $OPT({\II})$. Thus, our algorithm initially enumerates over all possible partial assignments of constant size. Each assignment is then extended to an approximate solution for ${\II}$. Among all possible partial assignments and the respective extensions the algorithm returns the best solution. Thus, from now on we restrict our attention to finding a solution for the {\em residual} problem, obtained by fixing the initial partial assignment.

Formally, given an SMKP instance, ${\II}=(I,w,B,W,f)$, a feasible partial solution $(A_b)_{b\in B}$ and $\xi \in \mathbb{N}$, we define the {\em residual} instance ${\II}'=(I',w,B,W',f')$ with respect to $(A_b)_{b\in B}$ and $\xi$  as follows.
Let $A=\cup_{b\in B} A_b$
and set $I'= \left\{ i \in I\setminus A ~ \middle|~ f_A (\{i\})\leq \frac{f(A)}{\xi} \right\}$. The weights of the items remain the same and so is the  set of bins. 
For every $b\in B$ we set $W'_b = W_b-w(A_b)$.
Finally, the objective function of the residual instance is $f'=f_A$.
%
\begin{lemma}
	\label{lem:residual}
	Let $\II$ be an SMKP instance,  $\xi \in \mathbb{N}$, and $(A^*_b)_{b\in B}$ an optimal solution for ${\II}$ such that $A^*_{b_1} \cap A^*_{b_2}=\emptyset$ for any $b_1, b_2 \in B$, $b_1 \neq b_2$. If $\sum_{b\in B} |A^*_b| \geq \xi$ there is a feasible solution $(A_b)_{b\in B}$ for $\II$ such that $A_b\subseteq A^*_b$ for any $b\in B$,
	$\sum_{b\in B} |A_b|= \xi$, and
	$(A^*_b\setminus A_b)_{b\in B}$ is a feasible solution for the residual instance of $\II'$ w.r.t $(A_b)_{b\in B}$ and $\xi$.
\end{lemma}

\begin{proof}
	Let   $(A^*_b)_{b\in B}$ be an optimal solution for the SMKP instance.  Define $A^*=\cup_{b\in B} A^*_b$ and order the items of $A^*$ by their {\em marginal values}.
	That is, $A^* = \{a_1,\ldots, a_r\}$ where $f_{T_{\ell-1}}(\{a_\ell\}) = \max_{a\in A^*\setminus T_{\ell-1}} f_{T_{\ell-1}}(\{a\})$ with $T_{\ell} =\{a_1,\ldots, a_{\ell}\}$   for every $1\leq \ell \leq r$ (also, $T_0=\emptyset$). Define $(A_b)_{b \in B}$ by $A_b= A^*_b \cap \{a_1,\ldots, a_{\xi}\}$ for every $b\in B$ and $A=\cup_{b\in B} A_b$. We therefore have $A=\{a_1,\ldots , a_{\xi}\}$.
	
	For any $b\in B$, it holds that $w(A_b) \leq w(A^*_b) \leq W_b$, and thus $(A_b)_{b \in B}$  is a feasible solution for $\II$. Furthermore, for any $b \in B$ it holds that $A_b \subseteq A^*_b$ by definition. As the sets $(A^*_b)_{b\in B}$ are disjoint it follows that $\sum_{b\in B} |A_b| = \xi$.
	
	Let $\II'=(I',w,B,W',f')$ be the residual instance of $\II$ w.r.t $(A_b)_{b \in B}$  and $\xi$.  It remains to show that  $(A^*_b\setminus A_b)_{b\in B}$ is a feasible solution for $\II'$.
	For every $\xi<i \leq r$ and $1\leq \ell \leq \xi$ it holds that $f_A(\{a_{i}\}) \leq f_{T_{\ell -1}}(\{a_{i}\})\leq f_{T_{\ell-1}} (\{a_{\ell}\})$  where the first inequality follows from the  submodularity of $f$ and the second by the definition of $a_{\ell}$. Combining the last inequality with $f' = f_A$ we obtain,
	\[\xi \cdot f'(\{a_i\})=\xi \cdot f_A(\{a_{i}\}) \leq \sum_{\ell = 1}^{\xi} f_{T_{\ell -1}} (\{a_{\ell}\})=f(A)-f(\emptyset) \leq f(A).\] Thus, $a_{i}\in I'$, implying that $A^*_b\setminus A_b \subseteq I'$ for any $b\in B$. Furthermore, for any $b\in B$,
	\[w(A^*_b \setminus A_b) = w(A^*_b) -w(A_b) \leq W_b-w(A_b) =W'_b.\]
	It follows that  $(A_b^*\setminus A_b )_{b\in B}$ is a solution for the residual instance. \qed
\end{proof}

Next, we observe that instances of SMKP are easier to solve when the number of distinct bin capacities is small (e.g., uniform bin capacities), leading us to consider {\em bin blocks}:
\begin{definition}
	For a given instance of SMKP we say that a subset of bins $\tilde{B}\subseteq B$ is a {\em block} if all the bins in $\tilde{B}$ have the same capacity, i.e., for bins $b_1$ and $b_2$ belonging to the same block it holds that $W_{b_1}= W_{b_2}$.
\end{definition}

Following an enumeration over partial assignments, our algorithm reduces the number of blocks by altering the bin capacities.
To this end,  we use a specific structure that we call {\it leveled}, defined as follows.
\begin{definition}
	For any $N\in \mathbb{N}$, we say that a partition $(B_j)_{j=0}^k$ of
	a set $B$ of bins with capacities $(W_b)_{b\in B}$
	is {\em $N$-leveled} if $B_j$ is a block, and  $|B_j|= N^{\floor{\frac{j}{N^2}}}$ for all $0\leq j \leq k$.
\end{definition}

By the above definition, we can view each set of consecutive blocks of the same size as a {\em level}.
For $0 \leq j \leq k$, block $j$ belongs to level $\ell = \lfloor  \frac{j}{N^2}\rfloor$. 
Thus, for level $\ell>0$ the number of bins in each block of level $\ell$ is $N$ times the number of bins in each block of level $\ell -1$.

In Section \ref{sec:structuring} we give Algorithm \ref{alg:structuring} which generates an $N$-leveled partition of the bins, ${\tilde B}= \cup_{j=0}^k {\tilde B}_j$ with the capacities of the bins $(W_b)_{b\in B}$ modified to $(\tilde{W}_b)_{b\in \tilde{B}}$.
We show that solving the problem with these new bin capacities may cause only a small harm to the optimal solution value. In particular, we prove (in Section \ref{sec:structuring}) the following.

\begin{lemma}
	\label{lem:structuring}
	Algorithm \ref{alg:structuring} is a polynomial time algorithm which given 
	 $N\in \mathbb{N}$,  a set of bins $B$ and capacities $(W_b)_{b\in B} $, returns  a subset of bins $\tilde{B}\subseteq B$,  capacities $(\tilde{W}_b)_{b\in B} $,
	and an $N$-leveled partition  $(\tilde{B}_j)_{j=0}^k$ of $\tilde{B}$,
	such that
	\begin{enumerate}
		\item
		The bin capacities satisfy $\tilde{W}_b \leq W_b$, for every $b\in \tilde{B}$.
		\item
		\label{prop:struct_third}
		For any set of items $I$, weights $(w_i)_{i\in I}$, a  submodular non-negative function $f:2^I \rightarrow \mathbb{R}_{\geq 0}$, and a feasible assignment $(S_b)_{b\in B}$ for the instance $(I,w,B,W,f)$, there exists  a feasible assignment  $(\tilde{S}_b)_{b\in \tilde{B}}$ for the instance $(I,w,\tilde{B}, \tilde{W}, f)$ such that $f\left(\bigcup_{b\in \tilde{B}} \tilde{S}_b\right) \geq \left(1- \frac{1}{N}\right) f\left(\bigcup_{b\in B} S_b \right)$ and $\bigcup_{b\in \tilde{B}} \tilde{S}_b \subseteq \bigcup_{b\in B} S_b$.  
	\end{enumerate}
	We refer to $\tilde{B}$ and $\tilde{W}$ as the $N$-leveled constraint of $B$ and $W$. 
\end{lemma}

Once the instance is $N$-leveled, we proceed to solve the problem (fractionally) and apply randomized rounding  to obtain an integral solution (see Section~\ref{sec:rounding}).
Algorithm~\ref{alg:rounding} utilizes efficiently the leveled structure of the instance.  Instead of having a separate constraint for each bin in a block $-$
to bound the total size of the items packed in this bin $-$
we use only two constraints for each block. The first constraint is a knapsack constraint referring to the total capacity of a block, and the second constraint restricts the number of {\em configurations} assigned to the block.\footnote{We defined a configuration in Section \ref{sec:tools_techniques}.}
Thus, the number of constraints significantly decreases if the blocks are large. Since leveled instances also have a constant number of blocks consisting of a {\em single} bin, those are handled separately via the notion of $\delta$-restricted SMKP.
	
Given $\delta>0$, the input for $\delta$-{\em restricted SMKP} includes the same parameters as an input for SMKP, and also a subset $B^r \subseteq B$ of
{\em restricted} bins. A solution for $\delta$-restricted SMKP is a feasible assignment $(A_b)_{b\in B}$
satisfying also the property that $\forall b \in B^r$ the items assigned to $b$
are relatively small; namely,
for any  $b\in B^r$ and $i \in A_b$ it holds that  $w_i \leq \delta W_b$.

Given the $N$-leveled instance of our problem,
we turn the blocks of a single bin (that is, blocks $\tilde{B}_j$ such that $|\tilde{B}_j|=1)$ to be {\em restricted}.
We note that while items of  weight greater than $\delta W_b$ may be assigned to these blocks in some optimal solution, the overall number of such items is bounded by a constant. Indeed, our initial enumeration guarantees that evicting these items from an optimal
solution may cause only small harm to the optimal solution value, allowing us to consider the instance as $\delta$-restricted.

In Section~\ref{sec:rounding} we show the following bound on the performance guarantee of Algorithm~\ref{alg:rounding}, which uses randomized rounding. The algorithm is parameterized by $\aeps\in (0,0.1)$ (to be determined).
Suppose we are given a $\delta$-restricted SMKP instance $\II$, such that
the unrestricted bins are partitioned into blocks, i.e., $B\setminus B^r = B_1 \cup \ldots \cup B_k$, and
\begin{equation}
\label{eq:def_upsilon}
\upsilon =\max_{i\in I } f(\{i\}) - f(\emptyset).
\end{equation}
\begin{lemma}
	\label{lem:rounding}
	For $\aeps\in (0,0.1)$,
	Algorithm \ref{alg:rounding} returns a feasible solution $(S_b)_{b\in B}$ such that \\
	$\E\left[ f(\cup_{b\in B} S_b)\right] \geq (1-e^{-1})\frac{(1-\aeps)^3}{1+\aeps} (1-\gamma) \OPT(\II)$,
	where
	$$\gamma = \exp\left(-\frac{\aeps^3 }{16 }\cdot\frac{\OPT(\II)}{\upsilon} \right)+ |B^r| \exp\left(-\frac{\aeps^2}{12} \cdot \frac{1}{\delta}\right)  + 2 \cdot \sum_{j=1}^{k} \exp\left(-\frac{\aeps^2}{12} |B_j| \right).
	$$
\end{lemma}

Algorithm~\ref{alg:smkp} gives the pseudocode of our approximation algorithm for general SMKP instances.
The algorithm uses several configuration parameters that will be set in the proof of Lemma~\ref{lem:approx}.

\begin{algorithm}[h]
	\SetAlgoLined
	\SetKwInOut{Input}{Input}\SetKwInOut{Output}{output}
	\Input{An SMKP instance $\II =(I,w,B,W,f)$ and the parameters $N,\xi, \delta$ and $\mu$.}
	\DontPrintSemicolon

	\ForAll{feasible assignments $A=(A_b)_{b\in B}$ such that $\sum_{b\in B} |A_b|\leq \xi$}{
		
		\label{line:enum}
		
		Let $\II'=(I',w,B,W',f')$ be the residual instance of $\II$ w.r.t $(A_b)_{b\in B}$ and $\xi$. \;	
		
		Run Algorithm \ref{alg:structuring} with  the bins $B$ and capacities $(W'_b)_{b\in B}$.
		Let  $\tilde{B}$ and $(\tilde{W}_b)_{b\in \tilde{B}}$  be the output, and $\tilde{B} = \cup_{j=0}^{k} \tilde{B}_j$ the partition of $\tilde{B}$ to leveled blocks. Let $\tilde{\II}=(I',w,\tilde{B}, \tilde{W}, f')$ be the resulting instance.
		\label{line:structuring}
		\;
		
		Let $\tilde{\II}_R$ be the $\delta$-restricted SMKP instance of $\tilde{\II}$ with the restricted bins  $\tilde{B}^r = \cup_{j=0}^{\min\{N^2-1, k\}} \tilde{B}_j$.
		\label{line:restricted}
		\;
		
		Solve $\tilde{\II}_R$ using Algorithm \ref{alg:rounding} with parameter $\aeps$, and the partition $\tilde{B}\setminus \tilde{B}^r =\cup_{j=N^2}^{k}\tilde{B}_j $.  Denote the returned assignment by $(\tilde{S}_b)_{b\in \tilde{B}}$, and let $S_b = \tilde{S}_b$ for $b\in \tilde{B}$ and $S_b = \emptyset$ for $b\in B \setminus \tilde{B}$.\;
		
		If $f(\cup_{b\in B} (A_b \cup S_b))$ is higher than the value of the current best solution, set $(A_b\cup S_b)_{b\in B}$   as the current best solution.
		\label{line:solution}
		\;
	}
	Return the best solution found.
	\;
	\caption{Algorithm for SMKP}
	\label{alg:smkp}
\end{algorithm}

\begin{lemma}
	\label{lem:approx}
	For any $\eps>0$,
	there are parameters $N,\xi,\delta, \aeps$ such that,
	for any SMKP instance $\II$,
	Algorithm \ref{alg:smkp}
	returns a solution of expected value at least $(1-e^{-1}-\eps)\OPT(\II)$.
\end{lemma}

\begin{proof}
	
	We start by setting the parameter values. The reason for selecting these values will become clear later.
	Given a fixed $\eps\in (0,0.1)$, there is $\aeps\in (0,0.1)$ such that $\frac{(1-\aeps)^3}{1+\aeps} \geq (1-\eps^2)$.
	By the Monotone Convergence Theorem,
	$$ \lim_{N\rightarrow \infty}2 N^2 \cdot \sum_{t=1}^{\infty} \exp\left(-\frac{\aeps^2\cdot N^t}{12}\right) =   \sum_{t=1}^{\infty}
	\lim_{N\rightarrow \infty}2 N^2 \exp\left(-\frac{\aeps^2\cdot N^{t}}{12} \right) =0.$$
	It follows that there are $N>\frac{1}{\eps^2}$ and $\delta>0$ such that
	\begin{equation}
	\label{eq:Ndelta}
	N^2 \exp\left(-\frac{\aeps^2}{12} \cdot \frac{1}{\delta}\right)  + 2 N^2 \cdot \sum_{t=1}^{\infty} \exp\left(-\frac{\aeps^2}{12} N^t \right) <\frac{\eps^2}{2}.
	\end{equation}
	Finally, we select $\xi$ such that  $\xi\geq \frac{N^2}{\eps^2 \delta}$ and $\exp\left( -\frac{\mu^3}{16}  \cdot \frac{\xi}{5}\right)\leq \frac{\eps^2}{2}$.
	
	Let $\II=(I,w,B,W,f)$ be an SMKP instance, and let $(A^*_b)_{b\in B}$ be an optimal solution for ${\II}$. Assume w.l.o.g that $A^*_{b_1} \cap A^*_{b_2} =\emptyset$ for any $b_1, b_2 \in B$, $b_1\neq b_2$.
	Define $A^* = \cup_{b\in B} A^*_b$. If $|A^*|\leq \xi$, there is an iteration of Line \ref{line:enum} in which $A^*_b =A_b$ for all $b\in B$. Therefore, in this iteration we have at Line \ref{line:solution} $f(\cup_{b\in B} (A_b \cup S_b))\geq f(A^*)$, and the algorithm returns a solution of value at least $f(A^*)$.  Otherwise, by Lemma \ref{lem:residual}, there is a feasible solution $(A_b)_{b\in B}$ such that $A_b\subseteq A^*_b$, $\sum_{b\in B} |A_b| =\xi$ and $(A^*_b\setminus A_b)_{b\in B}$ is a feasible solution for $\II'$, the residual instance of $\II$ w.r.t $(A_b)_{b\in B}$ and $\xi$. It follows that there is an iteration of Line \ref{line:enum} which considers this solution $(A_b)_{b\in B}$. We focus on this iteration for the rest of the analysis.
	
	Let $A=\cup_{b\in B} A_b$. If $f(A)\geq (1-e^{-1}) f(A^*)$ then when the algorithm reaches Line \ref{line:solution} it holds that $f(\cup_{b\in B} (A_b \cup S_b))\geq f(A)\geq (1-e^{-1})f(A^*)$; therefore, the algorithm returns a $(1-e^{-1})$-approximation in this case.
	Henceforth, we can assume that $f(A)\leq (1-e^{-1})f(A^*)$. Then,
	\begin{equation*}
	f'(\cup_{b\in B}(A^*_b\setminus A_b)) =
	f'(A^*\setminus A) =
	f(A^*) -f(A)
	\geq \frac{f(A)}{1-e^{-1} }-f(A) =
	\frac{1}{e-1} f(A).
	\end{equation*}

	Since $(A^*_b\setminus A_b)_{b\in B}$ is a feasible solution for $\II'$, and by Lemma \ref{lem:structuring}, it holds that
	\begin{equation}
	\label{eq:tII_bound}
	\OPT(\tilde{\II} )\geq \left( 1-\frac{1}{N}\right)f'(A^*\setminus A) \geq (1-\eps^2) f'(A^*\setminus A),
	\end{equation}
	where $\tilde{\II}$ is the instance defined in Step \ref{line:structuring}. The last inequality follows from the definition of $N$. 
	Let $(D_b)_{b\in \tilde{B}}$ be an optimal solution for $\tilde{\II}$.
	Consider $(D^r_b)_{b\in \tilde{B}}$ where
	$D^r_b = D_b\setminus \{ i\in D_b | w_i >\delta\cdot \tilde{W}_b \}$ for $b\in \tilde{B}^r$  (the set $\tilde{B}^r$ is defined in Line \ref{line:restricted}) and
	$D^r_b = D_b$ for $b\in \tilde{B} \setminus \tilde{B}^r$. It follows that $D^r_b$ is a solution for the $\delta$-restricted SMKP instance $\tilde{\II}_R$.
	As for any $b\in \tilde{B}^r$ it holds that $|\{i\in D_b | w_i >\delta \cdot \tilde{W}_b\} | \leq \frac{1}{\delta}$, we have
	\begin{equation}
	\label{eq:I_R_lb}
	\begin{aligned}
	\OPT\left(\tilde{\II}_R\right)&\geq f'\left(\cup_{b \in \tilde{B}} D^r_b\right)
	\geq \OPT\left(\tilde{\II}\right)-
	\frac{N^2}{\delta\cdot \xi } f(A) \geq (1-\eps^2) f'\left(A^*\setminus A\right) - \eps^2\cdot f(A).
	\end{aligned}
	\end{equation}
	The second inequality follows from the definition of residual instance, and the third inequality from \eqref{eq:tII_bound} and the choice of $\xi$.
	Since $f'(A^*\setminus A)\geq  \frac{1}{e-1}f(A)$ and $\eps \in (0,0.1)$, it follows that $\OPT(\tilde{\II}_R) \geq \frac{f(A)}{5}$.
	
	Let $\upsilon = \max_{i\in I'} f'(\{i\})$. By Lemma \ref{lem:rounding}, we have that
	\begin{equation}
	\label{eq:E_f_S}
	\E\left[f'(\cup_{b\in \tilde{B}} \tilde{S}_b) \right] \geq (1-e^{-1})\frac{(1-\aeps)^3}{1+\aeps} (1-\gamma) \OPT(\tilde{\II}_R) \geq (1-e^{-1})(1-\eps^2)(1-\gamma)\OPT(\tilde{\II}_R),
	\end{equation}
	where
	\begin{eqnarray}
	\nonumber
	\gamma &= 
	\exp\left(-\frac{\aeps^3 }{16 }\cdot\frac{\OPT(\tilde{\II}_{R})}{\upsilon} \right)+ |\tilde{B}^r| \exp\left(-\frac{\aeps^2}{12} \cdot \frac{1}{\delta}\right)  + 2 \cdot \sum_{j=N^2}^{k} \exp\left(-\frac{\aeps^2}{12} |\tilde{B}_j| \right)\\
	&
	\leq\exp\left(-\frac{\aeps^3 }{16 }\cdot\frac{\OPT(\tilde{\II}_R)}{\xi^{-1} f(A)} \right)+ |\tilde{B}^r| \exp\left(-\frac{\aeps^2}{12} \cdot \frac{1}{\delta}\right)  + 2 \cdot \sum_{j=N^2}^{k} \exp\left(-\frac{\aeps^2}{12} |\tilde{B}_j| \right) \nonumber \\
	\label{eq:bound_mu}
	&\leq \exp\left(-\frac{\aeps^3 }{16 }\cdot\frac{\xi}{5} \right)+ N^2 \exp\left(-\frac{\aeps^2}{12} \cdot \frac{1}{\delta}\right)  + 2\cdot N^2 \cdot \sum_{t=1}^{\infty} \exp\left(-\frac{\aeps^2}{12} N^t \right)\leq  \eps^2.
	\end{eqnarray}
	
	The first inequality uses $\upsilon = \max_{i\in I'} f'(\{i\})\leq \xi^{-1} f(A)$  (by the definition of $\II'$).
	The second inequality holds since $\OPT(\tilde{\II}_R)\geq \frac{f(A)}{5}$ , $|\tilde{B}^r|\leq N^2$ 
	and there are at most $N^2$ blocks $\tilde{B}_j$ of size $N^t$. The last inequality uses 
	\eqref{eq:Ndelta} and the choice of $\xi$.  Combining \eqref{eq:bound_mu} with \eqref{eq:E_f_S} and \eqref{eq:I_R_lb}, we obtain
	\begin{equation*}
	\begin{aligned}
	&\E\left[f(\cup_{b\in B} (A_b\cup S_b)) \right]
	\geq
	f(A)+\E\left[f'(\cup_{b\in \tilde{B}} \tilde{S}_b) \right]  \geq f(A)+ (1-e^{-1})(1-\eps^2)^2\OPT(\tilde{\II}_R) \\
	\geq &f(A)+ (1-e^{-1})(1-\eps^2)^3 f'(A^*\setminus A)-\eps^2 f(A)\geq (1-e^{-1}-\eps) f(A^*).
	\end{aligned}
	\end{equation*}
	
	Hence, in this iteration the solution considered in Line \ref{line:solution} has expected value at least $(1-e^{-1}-\eps) f(A^*)$. This completes the proof of the lemma. \qed
	\end{proof}

\begin{lemma}
	\label{lem:correct}
	For any constant parameters $N$, $\xi$, $\delta$ and $\aeps$,  Algorithm \ref{alg:smkp} returns a feasible solution for the input instance in polynomial time.
\end{lemma}

\begin{proof}
	We first note that for any fixed parameter values the algorithm has a polynomial running time. The number of assignments considered in Line \ref{line:enum} can be trivially bounded by $(n\cdot m )^\xi$.
	As Algorithms \ref{alg:structuring}  and \ref{alg:rounding} are polynomial in their input size, the operations in each iteration are also done in polynomial time.
	
	For each iteration of Line \ref{line:enum}, by Lemma \ref{lem:rounding}, $(\tilde{S}_b)_{b\in \tilde{B}}$ is a feasible solution to $\tilde{\II}_R$. Therefore, for any $b\in B$ either $w(S_b)=w(\emptyset)\leq W'_b$ or $w(S_b)=w(\tilde{S}_b)\leq \tilde{W}_b\leq W'_b$, where the last equality follows from Lemma \ref{lem:structuring}. Therefore, $w(A_b\cup S_b)\leq w(A_b) + W'_b \leq W_b$. Hence, the solution considered in each iteration is feasible for the input instance.
	\qed
\end{proof}

Theorem \ref{thm:main_result} follows 
from Lemmas \ref{lem:approx} and \ref{lem:correct}.

\subsection{Structuring the Instance}\label{sec:structuring}
In this section we present Algorithm \ref{alg:structuring} and prove Lemma \ref{lem:structuring}.
Our technique for generating an $N$-leveled partition can be viewed as a variant of the linear grouping technique of~\cite{de1981bin}.
We start with a brief overview of the classical concepts of grouping and shifting in the context of a multiple knapsack constraint.

Let $B=\{1,2,\ldots, m\}$ be a  set of bins with capacities $(W_b)_{b\in B}$, where
$W_1 \geq W_2 \geq \ldots \geq W_m$ and $m=q\cdot N^2$ for some integer $q \geq 1$.
We can partition $B$ into $N^2$ groups (sets) $B_1, \ldots ,B_{N^2}$, each consists of $q$ consecutive bins, i.e., 
$B_j=\{(j-1)\cdot q +1,\ldots, j\cdot q\}$ for $1\leq j \leq N^2$. Thus, the capacity of a bin in $B_{j}$ is greater or equal to the capacity of a bin in $B_{j+1}$. 

 We use the partition to define  new capacities for the bins. The  new capacity of a bin $b\in B_j$ is $\tilde{W}_{b}= \min_{b'\in B_j} W_{b'}=W_{q\cdot j}$, the  minimal (original) capacity of a bin in its group. Clearly,
 given an SMKP instance $\II=(I,w,B,W,f)$  and a feasible assignment $(S_b)_{b\in B}$ for the instance, it may be that $(S_b)_{b\in B}$ is infeasible for the instance with  the new capacities $\tilde{\II}=(I,w,B,\tilde{W},f)$.
 
 We can apply {\em shifting} to partially circumvent this hurdle. Given $b\in B_{j+1}$, $j\neq N^2$, the set $S_b$ complies with  the new capacity constraint of any bin $b'\in B_{j}$, i.e., $w(S_b)
\leq W_b\leq \tilde{W}_{b'}$. Define a new assignment $(\tilde{S}_b)_{b\in B}$  by $\tilde{S}_b=S_{b+q}$ for $b\in B\setminus B_{N^2}$
and $\tilde{S}_b=\emptyset$ for $b\in B_{N^2}$. 
As  $b+q\in B_{j+1}$  for any $b\in B_j$, $j\neq N^2$, it follows that $(\tilde{S}_b)_{b\in B}$ 
 is a feasible assignment for $\tilde{\II}$  . Furthermore, $(\tilde{S}_b)_{b\in B}$ is an assignment of all the items in $(S_b)_{b\in B}$, except for the items $\bigcup_{b\in B_1} S_b$ assigned to the first group in $(S_b)_{b\in B}$.
 The assignment of these items is handled by different techniques.
 
%
%
%
%

Algorithm \ref{alg:structuring} applies grouping with {\em non-uniform} group size 
to generate the $N$-leveled partition.
The  algorithm assumes w.l.o.g that the set of bins is $B=\{1,2,\ldots, m\}$ and that the bins are ordered by capacity, $W_1\geq W_2\geq \ldots \geq W_m$. 
It  defines groups
(or blocks) of bins, where  group $j$ consists of $N^{\lfloor  \frac{j}{N^2} \rfloor}$ consecutive bins, for $j \geq 0$. 
The capacity of the bins in each group is reduced to the minimal capacity of a bin in this group. 
This procedure is formalized in Algorithm~\ref{alg:structuring}. A simple illustration for a small instance is given in Figure \ref{fig:leveling}.

\begin{figure}
	\centering
	\caption[position=bottom]{Input and output example for Algorithm \ref{alg:structuring} with $N=2$. 
	 The original capacities, $W$, are represented by empty rectangles, whereas the hatched rectangles represent the new capacities $\tilde{W}$. Note that the last three bins are discarded by the algorithm as they do not form a full block.}
	\label{fig:leveling}

	\begin{tikzpicture}
	\draw [decorate,decoration={brace,amplitude=10pt}] (2.2,-0.1) -- (-0.1,-0.1) node [black,midway,yshift=-0.5cm]
	{\footnotesize Level $0$};
	\draw[pattern=north west lines] (0,0) rectangle ++(0.3,4);
	\draw[pattern=north west lines] (0.6,0) rectangle ++(0.3,3.8);
	\draw[pattern=north west lines] (1.2,0) rectangle ++(0.3,3.7);
	\draw[pattern=north west lines] (1.8,0) rectangle ++(0.3,3.4);
	
	\draw [decorate,decoration={brace,amplitude=10pt}] (6.4,-0.1) -- (2.5,-0.1) node [black,midway,yshift=-0.5cm]
	{\footnotesize Level $1$};
	\draw[] (2.6,0) rectangle ++(0.3,3.4);
	\draw[pattern=north west lines] (2.6,0) rectangle ++(0.3,3);
	\draw[pattern=north west lines] (3.0,0) rectangle ++(0.3,3);
	
	\draw[] (3.6,0) rectangle ++(0.3,2.8);
	\draw[pattern=north west lines] (3.6,0) rectangle ++(0.3,2.4);
	\draw[pattern=north west lines] (4.0,0) rectangle ++(0.3,2.4);
	
	\draw[] (4.6,0) rectangle ++(0.3,1.8);
	\draw[pattern=north west lines] (4.6,0) rectangle ++(0.3,1.8);
	\draw[pattern=north west lines] (5.0,0) rectangle ++(0.3,1.8);
	
	\draw[] (5.6,0) rectangle ++(0.3,1.7);
	\draw[pattern=north west lines] (5.6,0) rectangle ++(0.3,1.2);
	\draw[pattern=north west lines] (6.0,0) rectangle ++(0.3,1.2);
	
	\draw [decorate,decoration={brace,amplitude=10pt}] (8.4,-0.1) -- (6.7,-0.1) node [black,midway,yshift=-0.5cm]
	{\footnotesize Level $2$};
	\draw[] (6.8,0) rectangle ++(0.3,1);
	\draw[pattern=north west lines] (6.8,0) rectangle ++(0.3,0.5);
	\draw[] (7.2,0) rectangle ++(0.3,0.8);
	\draw[pattern=north west lines] (7.2,0) rectangle ++(0.3,0.5);
	\draw[] (7.6,0) rectangle ++(0.3,0.8);
	\draw[pattern=north west lines] (7.6,0) rectangle ++(0.3,0.5);
	\draw[pattern=north west lines] (8.0,0) rectangle ++(0.3,0.5);
	\draw[] (8.6,0) rectangle ++(0.3,0.5);
	\draw[] (9.0,0) rectangle ++(0.3,0.5);
	\draw[] (9.4,0) rectangle ++(0.3,0.4);
	
	\end{tikzpicture}
\end{figure}
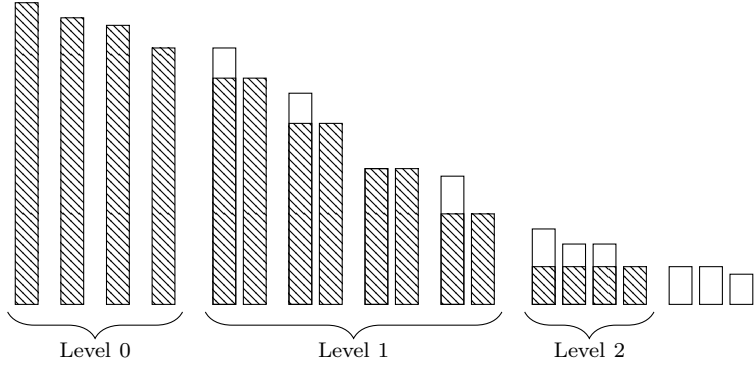

\begin{algorithm}[H]
	\SetAlgoLined
	\SetKwInOut{Input}{Input}\SetKwInOut{Configuration}{Configuration}
	\DontPrintSemicolon
	
	\Input{A set of bins $B$,  capacities $(W_b)_{b\in B}$ and $N\in \mathbb{N}$.}

	Let $B=\{1,\ldots, m\}$ where $W_1\geq W_2 \geq \ldots \geq  W_m$.\;
	
	Let $k = \max \left\{\ell \in \mathbb{N} ~\middle|~\sum_{r=0}^{\ell} N^{\floor{\frac{r}{N^2}}} \leq m\right\}$. \;
	
	Define $\tilde{B}_{j}= \left\{ b~\middle|~
	\sum_{r=0}^{j-1} N^{\floor{\frac{r}{N^2}}} ~< ~b~\leq \sum_{r=0}^{j} N^{\floor{\frac{r}{N^2}}} 
	\right\}$  for $0\leq j \leq k$.\; 
	
	Let $\tilde{B}= \bigcup_{j=0}^k \tilde{B}_j$, and $\tilde{W}_b= \min_{b'\in \tilde{B}_j} W_{b'}$ for all $0\leq j \leq k$ and $b\in \tilde{B}_j$.\;
	
	Return $\tilde{B}$, $(\tilde{W}_b)_{b\in \tilde{B}}$ and the partition $(\tilde{B}_j)_{j=0}^k$. 
	
	\caption{Structure in Blocks}
	\label{alg:structuring}
\end{algorithm}

By construction, we have that $(\tilde{B}_j)_{j=0}^{k}$ is an $N$-leveled partition of  $\tilde{B}$. Furthermore, $\tilde{B}\subseteq B$ and  $\tilde{W}_b \leq W_b$ for any $b\in \tilde{B}$. Finally, it can be  easily observed that  Algorithm \ref{alg:structuring} has a polynomial running time. 
Thus, to complete the proof of Lemma \ref{lem:structuring} we need to show that property \ref{prop:struct_third} holds as well. To this end, we use a variant of the shifting argument outlined in the above overview.

\begin{lemma}
	\label{lem:structuring_internal}
	Let $N\in \mathbb{N}$, $B$ be a set of bins with capacities $(W_b)_{b\in B}$ and let $\tilde{B}$, $\tilde{W}$ be the output of Algorithm \ref{alg:structuring}  for the input $B$, $W$ and $N$.
	Furthermore, let $I$ be a set of items with weights $(w_i)_{i\in I}$, $f:2^I \rightarrow \mathbb{R}_{\geq 0}$ be a  submodular non-negative function,
	 and  $(S_b)_{b\in B}$  be a feasible assignment for
	 $(I,w,B,W,f)$. Then there is $(\tilde{S}_b)_{b\in \tilde{B}}$ feasible for $(I,w,\tilde{B}, \tilde{W},f)$ such that
	 	 $f\left(\bigcup_{b\in \tilde{B}}\tilde{S}_b\right)\geq \left( 1- \frac{1}{N}\right)f\left(\bigcup_{b\in B} S_b\right)$ and $\bigcup_{b\in \tilde{B} } \tilde{S}_b \subseteq \bigcup_{b\in B} S_b$.
\end{lemma}
Clearly, Lemma \ref{lem:structuring_internal} completes the proof of Lemma \ref{lem:structuring}.
To prove Lemma \ref{lem:structuring_internal} we use the following property of submodular functions.
\begin{lemma}
	\label{lem:simple_evict}
	Let $h:2^{I}\rightarrow \mathbb{R}$ be a submodular  function and let $S_1,\ldots, S_N \subseteq I$ be disjoint sets. Then there is $1\leq r^* \leq N$ such that 
	$$h\left(\bigcup_{~1\leq r \leq N,~r\neq r^*~} S_r\right)~\geq~ 
	\left(1-\frac{1}{N} \right) h(S_1 \cup \ldots \cup S_N) .$$ 
\end{lemma} 

The proof of the Lemma \ref{lem:simple_evict} is given in  \ref{app:submodular}.
\newproof{pstruct}{Proof of Lemma \ref{lem:structuring_internal}}

\begin{figure}
	\centering
	\caption[bottom]{Illustration of the steps in the proof of Lemma~\ref{lem:structuring_internal}
		for $N=3$. Each row represents a level and each box represents either a block or a super-block. The number in the box is the number of bins in the block and a gray background implies that the block is non-empty.}
	\label{fig:shifting}
	
	\begin{tikzpicture}
	\tikzmath{
		\b1=0;\b2=0.5;\b3=1.0;\b4=1.5;\b5=2.0;\b6=2.5;\b7=3.0;\b8=3.5;\b9=4.0;
		\yzero=2.4;\yone=1.4;\ytwo=0.4;
		\rightsidedelta=8.5;\bottomdelta=6;
	}
	\draw [->] (\b1-0.5,\yzero+\bottomdelta+0.25) to (\b1-0.2,\yzero+\bottomdelta+0.25) node [black,xshift=-0.6cm] {\footnotesize $\mL_0$};
	\draw [<->,Bar-Bar] (\b1,\yzero+\bottomdelta+0.65) to (\b4,\yzero+\bottomdelta+0.65) node [black,xshift=-0.7cm,yshift=0.05cm,fill=white] {\footnotesize $\mS_{0,0}$};
	\draw [<->,Bar-Bar] (\b4,\yzero+\bottomdelta+0.65) to (\b7,\yzero+\bottomdelta+0.65) node [black,xshift=-0.7cm,yshift=0.05cm,fill=white] {\footnotesize $\mS_{0,1}$};
	\draw [<->,Bar-Bar] (\b7,\yzero+\bottomdelta+0.65) to (\b9+0.5,\yzero+\bottomdelta+0.65) node [black,xshift=-0.7cm,yshift=0.05cm,fill=white] {\footnotesize $\mS_{0,2}$};
	
	\filldraw[color=black, fill=gray!30] (\b1,\yzero+\bottomdelta) rectangle ++(0.5,0.5) node [black,midway,font=\bfseries,yshift=0cm] {\footnotesize $1$};
	\filldraw[color=black, fill=gray!30] (\b2,\yzero+\bottomdelta) rectangle ++(0.5,0.5) node [black,midway,font=\bfseries,yshift=0cm] {\footnotesize $1$};
	\filldraw[color=black, fill=gray!30] (\b3,\yzero+\bottomdelta) rectangle ++(0.5,0.5) node [black,midway,font=\bfseries,yshift=0cm] {\footnotesize $1$};
	\filldraw[color=black, fill=gray!30] (\b4,\yzero+\bottomdelta) rectangle ++(0.5,0.5) node [black,midway,font=\bfseries,yshift=0cm] {\footnotesize $1$};
	\filldraw[color=black, fill=gray!30] (\b5,\yzero+\bottomdelta) rectangle ++(0.5,0.5) node [black,midway,font=\bfseries,yshift=0cm] {\footnotesize $1$};
	\filldraw[color=black, fill=gray!30] (\b6,\yzero+\bottomdelta) rectangle ++(0.5,0.5) node [black,midway,font=\bfseries,yshift=0cm] {\footnotesize $1$};
	\filldraw[color=black, fill=gray!30] (\b7,\yzero+\bottomdelta) rectangle ++(0.5,0.5) node [black,midway,font=\bfseries,yshift=0cm] {\footnotesize $1$};
	\filldraw[color=black, fill=gray!30] (\b8,\yzero+\bottomdelta) rectangle ++(0.5,0.5) node [black,midway,font=\bfseries,yshift=0cm] {\footnotesize $1$};
	\filldraw[color=black, fill=gray!30] (\b9,\yzero+\bottomdelta) rectangle ++(0.5,0.5) node [black,midway,font=\bfseries,yshift=0cm] {\footnotesize $1$};
	
	\draw [->] (\b1-0.5,\yone+\bottomdelta+0.25) to (\b1-0.2,\yone+\bottomdelta+0.25) node [black,xshift=-0.6cm] {\footnotesize $\mL_1$};
	\draw [<->,Bar-Bar] (\b1,\yone+\bottomdelta+0.65) to (\b4,\yone+\bottomdelta+0.65) node [black,xshift=-0.7cm,yshift=0.05cm,fill=white] {\footnotesize $\mS_{1,0}$};
	\draw [<->,Bar-Bar] (\b4,\yone+\bottomdelta+0.65) to (\b7,\yone+\bottomdelta+0.65) node [black,xshift=-0.7cm,yshift=0.05cm,fill=white] {\footnotesize $\mS_{1,1}$};
	\draw [<->,Bar-Bar] (\b7,\yone+\bottomdelta+0.65) to (\b9+0.5,\yone+\bottomdelta+0.65) node [black,xshift=-0.7cm,yshift=0.05cm,fill=white] {\footnotesize $\mS_{1,2}$};
	
	\filldraw[color=black, fill=gray!30] (\b1,\yone+\bottomdelta) rectangle ++(0.5,0.5) node [black,midway,font=\bfseries,yshift=0cm] {\footnotesize $3$};
	\filldraw[color=black, fill=gray!30] (\b2,\yone+\bottomdelta) rectangle ++(0.5,0.5) node [black,midway,font=\bfseries,yshift=0cm] {\footnotesize $3$};
	\filldraw[color=black, fill=gray!30] (\b3,\yone+\bottomdelta) rectangle ++(0.5,0.5) node [black,midway,font=\bfseries,yshift=0cm] {\footnotesize $3$};
	\filldraw[color=black, fill=gray!30] (\b4,\yone+\bottomdelta) rectangle ++(0.5,0.5) node [black,midway,font=\bfseries,yshift=0cm] {\footnotesize $3$};
	\filldraw[color=black, fill=gray!30] (\b5,\yone+\bottomdelta) rectangle ++(0.5,0.5) node [black,midway,font=\bfseries,yshift=0cm] {\footnotesize $3$};
	\filldraw[color=black, fill=gray!30] (\b6,\yone+\bottomdelta) rectangle ++(0.5,0.5) node [black,midway,font=\bfseries,yshift=0cm] {\footnotesize $3$};
	\filldraw[color=black, fill=gray!30] (\b7,\yone+\bottomdelta) rectangle ++(0.5,0.5) node [black,midway,font=\bfseries,yshift=0cm] {\footnotesize $3$};
	\filldraw[color=black, fill=gray!30] (\b8,\yone+\bottomdelta) rectangle ++(0.5,0.5) node [black,midway,font=\bfseries,yshift=0cm] {\footnotesize $3$};
	\filldraw[color=black, fill=gray!30] (\b9,\yone+\bottomdelta) rectangle ++(0.5,0.5) node [black,midway,font=\bfseries,yshift=0cm] {\footnotesize $3$};
	
	\draw [->] (\b1-0.5,\ytwo+\bottomdelta+0.25) to (\b1-0.2,\ytwo+\bottomdelta+0.25) node [black,xshift=-0.6cm] {\footnotesize $\mL_2$};
	
	\filldraw[color=black, fill=gray!30] (\b1,\ytwo+\bottomdelta) rectangle ++(0.5,0.5) node [black,midway,font=\bfseries,yshift=0cm] {\footnotesize $9$};
	\filldraw[color=black, fill=gray!30] (\b2,\ytwo+\bottomdelta) rectangle ++(0.5,0.5) node [black,midway,font=\bfseries,yshift=0cm] {\footnotesize $9$};
	\filldraw[color=black, fill=gray!30] (\b3,\ytwo+\bottomdelta) rectangle ++(0.5,0.5) node [black,midway,font=\bfseries,yshift=0cm] {\footnotesize $9$};
	\filldraw[color=black, fill=gray!30] (\b4,\ytwo+\bottomdelta) rectangle ++(0.5,0.5) node [black,midway,font=\bfseries,yshift=0cm] {\footnotesize $9$};
	\filldraw[color=black, fill=gray!30] (\b5,\ytwo+\bottomdelta) rectangle ++(0.5,0.5) node [black,midway,font=\bfseries,yshift=0cm] {\footnotesize $9$};
	\filldraw[color=black, fill=gray!30] (\b6,\ytwo+\bottomdelta) rectangle ++(0.5,0.5) node [black,midway,font=\bfseries,yshift=0cm] {\footnotesize $9$};
	\filldraw[color=black, fill=gray!30] (\b7,\ytwo+\bottomdelta) rectangle ++(0.5,0.5) node [black,midway,font=\bfseries,yshift=0cm] {\footnotesize $9$};
	\filldraw[color=black, fill=gray!30] (\b8,\ytwo+\bottomdelta) rectangle ++(0.5,0.5) node [black,midway,font=\bfseries,yshift=0cm] {\footnotesize $9$};
	\filldraw[color=black, fill=gray!30] (\b9,\ytwo+\bottomdelta) rectangle ++(0.5,0.5) node [black,midway,font=\bfseries,yshift=0cm] {\footnotesize $9$};
	
	\draw[black, arrows={-Classical TikZ Rightarrow[angle=90:10pt,black]}]  (\b9+1.5,\yone+\bottomdelta+0.25) -- (\b9+3.5,\yone+\bottomdelta+0.25) node [black,midway,font=\bfseries,xshift=-0.1cm,yshift=0.4cm] {\footnotesize {\bf Eviction}};
	
	\filldraw[color=black, fill=gray!30] (\b1+\rightsidedelta,\yzero+\bottomdelta) rectangle ++(0.5,0.5) node [black,midway,font=\bfseries,yshift=0cm] {\footnotesize $1$};
	\filldraw[color=black, fill=gray!30] (\b2+\rightsidedelta,\yzero+\bottomdelta) rectangle ++(0.5,0.5) node [black,midway,font=\bfseries,yshift=0cm] {\footnotesize $1$};
	\filldraw[color=black, fill=gray!30] (\b3+\rightsidedelta,\yzero+\bottomdelta) rectangle ++(0.5,0.5) node [black,midway,font=\bfseries,yshift=0cm] {\footnotesize $1$};
	\filldraw[color=black,fill=gray!0] (\b4+\rightsidedelta,\yzero+\bottomdelta) rectangle ++(0.5,0.5) node [black,midway,font=\bfseries,yshift=0cm] {\footnotesize $1$};
	\filldraw[color=black, fill=gray!0] (\b5+\rightsidedelta,\yzero+\bottomdelta) rectangle ++(0.5,0.5) node [black,midway,font=\bfseries,yshift=0cm] {\footnotesize $1$};
	\filldraw[color=black, fill=gray!0] (\b6+\rightsidedelta,\yzero+\bottomdelta) rectangle ++(0.5,0.5) node [black,midway,font=\bfseries,yshift=0cm] {\footnotesize $1$};
	\filldraw[color=black, fill=gray!30] (\b7+\rightsidedelta,\yzero+\bottomdelta) rectangle ++(0.5,0.5) node [black,midway,font=\bfseries,yshift=0cm] {\footnotesize $1$};
	\filldraw[color=black, fill=gray!30] (\b8+\rightsidedelta,\yzero+\bottomdelta) rectangle ++(0.5,0.5) node [black,midway,font=\bfseries,yshift=0cm] {\footnotesize $1$};
	\filldraw[color=black, fill=gray!30] (\b9+\rightsidedelta,\yzero+\bottomdelta) rectangle ++(0.5,0.5) node [black,midway,font=\bfseries,yshift=0cm] {\footnotesize $1$};
	
	\filldraw[color=black, fill=gray!30] (\b1+\rightsidedelta,\yone+\bottomdelta) rectangle ++(0.5,0.5) node [black,midway,font=\bfseries,yshift=0cm] {\footnotesize $3$};
	\filldraw[color=black, fill=gray!30] (\b2+\rightsidedelta,\yone+\bottomdelta) rectangle ++(0.5,0.5) node [black,midway,font=\bfseries,yshift=0cm] {\footnotesize $3$};
	\filldraw[color=black, fill=gray!30] (\b3+\rightsidedelta,\yone+\bottomdelta) rectangle ++(0.5,0.5) node [black,midway,font=\bfseries,yshift=0cm] {\footnotesize $3$};
	\filldraw[color=black, fill=gray!0] (\b4+\rightsidedelta,\yone+\bottomdelta) rectangle ++(0.5,0.5) node [black,midway,font=\bfseries,yshift=0cm] {\footnotesize $3$};
	\filldraw[color=black, fill=gray!0] (\b5+\rightsidedelta,\yone+\bottomdelta) rectangle ++(0.5,0.5) node [black,midway,font=\bfseries,yshift=0cm] {\footnotesize $3$};
	\filldraw[color=black, fill=gray!0] (\b6+\rightsidedelta,\yone+\bottomdelta) rectangle ++(0.5,0.5) node [black,midway,font=\bfseries,yshift=0cm] {\footnotesize $3$};
	\filldraw[color=black, fill=gray!30] (\b7+\rightsidedelta,\yone+\bottomdelta) rectangle ++(0.5,0.5) node [black,midway,font=\bfseries,yshift=0cm] {\footnotesize $3$};
	\filldraw[color=black, fill=gray!30] (\b8+\rightsidedelta,\yone+\bottomdelta) rectangle ++(0.5,0.5) node [black,midway,font=\bfseries,yshift=0cm] {\footnotesize $3$};
	\filldraw[color=black, fill=gray!30] (\b9+\rightsidedelta,\yone+\bottomdelta) rectangle ++(0.5,0.5) node [black,midway,font=\bfseries,yshift=0cm] {\footnotesize $3$};
	
	\filldraw[color=black, fill=gray!30] (\b1+\rightsidedelta,\ytwo+\bottomdelta) rectangle ++(0.5,0.5) node [black,midway,font=\bfseries,yshift=0cm] {\footnotesize $9$};
	\filldraw[color=black, fill=gray!30] (\b2+\rightsidedelta,\ytwo+\bottomdelta) rectangle ++(0.5,0.5) node [black,midway,font=\bfseries,yshift=0cm] {\footnotesize $9$};
	\filldraw[color=black, fill=gray!30] (\b3+\rightsidedelta,\ytwo+\bottomdelta) rectangle ++(0.5,0.5) node [black,midway,font=\bfseries,yshift=0cm] {\footnotesize $9$};
	\filldraw[color=black, fill=gray!30] (\b4+\rightsidedelta,\ytwo+\bottomdelta) rectangle ++(0.5,0.5) node [black,midway,font=\bfseries,yshift=0cm] {\footnotesize $9$};
	\filldraw[color=black, fill=gray!30] (\b5+\rightsidedelta,\ytwo+\bottomdelta) rectangle ++(0.5,0.5) node [black,midway,font=\bfseries,yshift=0cm] {\footnotesize $9$};
	\filldraw[color=black, fill=gray!30] (\b6+\rightsidedelta,\ytwo+\bottomdelta) rectangle ++(0.5,0.5) node [black,midway,font=\bfseries,yshift=0cm] {\footnotesize $9$};
	\filldraw[color=black, fill=gray!30] (\b7+\rightsidedelta,\ytwo+\bottomdelta) rectangle ++(0.5,0.5) node [black,midway,font=\bfseries,yshift=0cm] {\footnotesize $9$};
	\filldraw[color=black, fill=gray!30] (\b8+\rightsidedelta,\ytwo+\bottomdelta) rectangle ++(0.5,0.5) node [black,midway,font=\bfseries,yshift=0cm] {\footnotesize $9$};
	\filldraw[color=black, fill=gray!30] (\b9+\rightsidedelta,\ytwo+\bottomdelta) rectangle ++(0.5,0.5) node [black,midway,font=\bfseries,yshift=0cm] {\footnotesize $9$};
	
	\draw[black, arrows={-Classical TikZ Rightarrow[angle=90:10pt,black]}]  (\b9+3.5,\ytwo+\bottomdelta-1.0) -- (\b9+1.5,\yzero+1.5) node [black,sloped, midway,font=\bfseries,xshift=0cm,yshift=0.4cm] {\footnotesize {\bf Shuffling}};

	\filldraw[color=black, fill=gray!30] (\b1,\yzero) rectangle ++(0.5,0.5) node [black,midway,font=\bfseries,yshift=0cm] {\footnotesize $1$};
	\filldraw[color=black, fill=gray!30] (\b2,\yzero) rectangle ++(0.5,0.5) node [black,midway,font=\bfseries,yshift=0cm] {\footnotesize $1$};
	\filldraw[color=black, fill=gray!30] (\b3,\yzero) rectangle ++(0.5,0.5) node [black,midway,font=\bfseries,yshift=0cm] {\footnotesize $1$};
	\filldraw[color=black, fill=gray!30] (\b4,\yzero) rectangle ++(0.5,0.5) node [black,midway,font=\bfseries,yshift=0cm] {\footnotesize $1$};
	\filldraw[color=black, fill=gray!30] (\b5,\yzero) rectangle ++(0.5,0.5) node [black,midway,font=\bfseries,yshift=0cm] {\footnotesize $1$};
	\filldraw[color=black, fill=gray!30] (\b6,\yzero) rectangle ++(0.5,0.5) node [black,midway,font=\bfseries,yshift=0cm] {\footnotesize $1$};
	\filldraw[color=black, fill=gray!0] (\b7,\yzero) rectangle ++(0.5,0.5) node [black,midway,font=\bfseries,yshift=0cm] {\footnotesize $1$};
	\filldraw[color=black, fill=gray!0] (\b8,\yzero) rectangle ++(0.5,0.5) node [black,midway,font=\bfseries,yshift=0cm] {\footnotesize $1$};
	\filldraw[color=black, fill=gray!0] (\b9,\yzero) rectangle ++(0.5,0.5) node [black,midway,font=\bfseries,yshift=0cm] {\footnotesize $1$};
	
	\draw [<->] (\b5+0.25,\yzero+0.6) to [out=30,in=150] (\b8+0.25,\yzero+0.6);
	
	\filldraw[color=black, fill=gray!30] (\b1,\yone) rectangle ++(0.5,0.5) node [black,midway,font=\bfseries,yshift=0cm] {\footnotesize $3$};
	\filldraw[color=black, fill=gray!30] (\b2,\yone) rectangle ++(0.5,0.5) node [black,midway,font=\bfseries,yshift=0cm] {\footnotesize $3$};
	\filldraw[color=black, fill=gray!30] (\b3,\yone) rectangle ++(0.5,0.5) node [black,midway,font=\bfseries,yshift=0cm] {\footnotesize $3$};
	\filldraw[color=black, fill=gray!30] (\b4,\yone) rectangle ++(0.5,0.5) node [black,midway,font=\bfseries,yshift=0cm] {\footnotesize $3$};
	\filldraw[color=black, fill=gray!30] (\b5,\yone) rectangle ++(0.5,0.5) node [black,midway,font=\bfseries,yshift=0cm] {\footnotesize $3$};
	\filldraw[color=black, fill=gray!30] (\b6,\yone) rectangle ++(0.5,0.5) node [black,midway,font=\bfseries,yshift=0cm] {\footnotesize $3$};
	\filldraw[color=black, fill=gray!0] (\b7,\yone) rectangle ++(0.5,0.5) node [black,midway,font=\bfseries,yshift=0cm] {\footnotesize $3$};
	\filldraw[color=black, fill=gray!0] (\b8,\yone) rectangle ++(0.5,0.5) node [black,midway,font=\bfseries,yshift=0cm] {\footnotesize $3$};
	\filldraw[color=black, fill=gray!0] (\b9,\yone) rectangle ++(0.5,0.5) node [black,midway,font=\bfseries,yshift=0cm] {\footnotesize $3$};
	
	\draw [<->] (\b5+0.25,\yone+0.6) to [out=30,in=150] (\b8+0.25,\yone+0.6);
	
	\filldraw[color=black, fill=gray!30] (\b1,\ytwo) rectangle ++(0.5,0.5) node [black,midway,font=\bfseries,yshift=0cm] {\footnotesize $9$};
	\filldraw[color=black, fill=gray!30] (\b2,\ytwo) rectangle ++(0.5,0.5) node [black,midway,font=\bfseries,yshift=0cm] {\footnotesize $9$};
	\filldraw[color=black, fill=gray!30] (\b3,\ytwo) rectangle ++(0.5,0.5) node [black,midway,font=\bfseries,yshift=0cm] {\footnotesize $9$};
	\filldraw[color=black, fill=gray!30] (\b4,\ytwo) rectangle ++(0.5,0.5) node [black,midway,font=\bfseries,yshift=0cm] {\footnotesize $9$};
	\filldraw[color=black, fill=gray!30] (\b5,\ytwo) rectangle ++(0.5,0.5) node [black,midway,font=\bfseries,yshift=0cm] {\footnotesize $9$};
	\filldraw[color=black, fill=gray!30] (\b6,\ytwo) rectangle ++(0.5,0.5) node [black,midway,font=\bfseries,yshift=0cm] {\footnotesize $9$};
	\filldraw[color=black, fill=gray!30] (\b7,\ytwo) rectangle ++(0.5,0.5) node [black,midway,font=\bfseries,yshift=0cm] {\footnotesize $9$};
	\filldraw[color=black, fill=gray!30] (\b8,\ytwo) rectangle ++(0.5,0.5) node [black,midway,font=\bfseries,yshift=0cm] {\footnotesize $9$};
	\filldraw[color=black, fill=gray!30] (\b9,\ytwo) rectangle ++(0.5,0.5) node [black,midway,font=\bfseries,yshift=0cm] {\footnotesize $9$};

	\draw[black, arrows={-Classical TikZ Rightarrow[angle=90:10pt,black]}]  (\b9+1.5,\yone+0.25) -- (\b9+3.5,\yone+0.25) node [black,midway,font=\bfseries,xshift=-0.1cm,yshift=0.4cm] {\footnotesize {\bf Shifting}};

	\filldraw[color=black, fill=gray!30] (\b1+\rightsidedelta,\yzero) rectangle ++(0.5,0.5) node [black,midway,font=\bfseries,yshift=0cm] {\footnotesize $1$};
	\filldraw[color=black, fill=gray!30] (\b2+\rightsidedelta,\yzero) rectangle ++(0.5,0.5) node [black,midway,font=\bfseries,yshift=0cm] {\footnotesize $1$};
	\filldraw[color=black, fill=gray!30] (\b3+\rightsidedelta,\yzero) rectangle ++(0.5,0.5) node [black,midway,font=\bfseries,yshift=0cm] {\footnotesize $1$};
	\filldraw[color=black, fill=gray!30] (\b4+\rightsidedelta,\yzero) rectangle ++(0.5,0.5) node [black,midway,font=\bfseries,yshift=0cm] {\footnotesize $1$};
	\filldraw[color=black, fill=gray!30] (\b5+\rightsidedelta,\yzero) rectangle ++(0.5,0.5) node [black,midway,font=\bfseries,yshift=0cm] {\footnotesize $1$};
	\filldraw[color=black, fill=gray!30] (\b6+\rightsidedelta,\yzero) rectangle ++(0.5,0.5) node [black,midway,font=\bfseries,yshift=0cm] {\footnotesize $1$};
	\filldraw[color=black, fill=gray!30] (\b7+\rightsidedelta,\yzero) rectangle ++(1.5,0.5) node [black,midway,font=\bfseries,yshift=0cm] {\footnotesize $3$};
	
	\draw [->] (\b1+0.25+\rightsidedelta,\yone+0.5) to [out=30,in=-150] (\b8+0.25+\rightsidedelta,\yzero-0.02);
	
	\filldraw[color=black, fill=gray!30] (\b1+\rightsidedelta,\yone) rectangle ++(0.5,0.5) node [black,midway,font=\bfseries,yshift=0cm] {\footnotesize $3$};
	\filldraw[color=black, fill=gray!30] (\b2+\rightsidedelta,\yone) rectangle ++(0.5,0.5) node [black,midway,font=\bfseries,yshift=0cm] {\footnotesize $3$};
	\filldraw[color=black, fill=gray!30] (\b3+\rightsidedelta,\yone) rectangle ++(0.5,0.5) node [black,midway,font=\bfseries,yshift=0cm] {\footnotesize $3$};
	\filldraw[color=black, fill=gray!30] (\b4+\rightsidedelta,\yone) rectangle ++(0.5,0.5) node [black,midway,font=\bfseries,yshift=0cm] {\footnotesize $3$};
	\filldraw[color=black, fill=gray!30] (\b5+\rightsidedelta,\yone) rectangle ++(0.5,0.5) node [black,midway,font=\bfseries,yshift=0cm] {\footnotesize $3$};
	\filldraw[color=black, fill=gray!0] (\b6+\rightsidedelta,\yone) rectangle ++(0.5,0.5) node [black,midway,font=\bfseries,yshift=0cm] {\footnotesize $3$};
	\filldraw[color=black, fill=gray!30] (\b7+\rightsidedelta,\yone) rectangle ++(1.5,0.5) node [black,midway,font=\bfseries,yshift=0cm] {\footnotesize $9$};
	
	\draw [->,thick,-Triangle] (\b2+\rightsidedelta+0.13,\yone+0.25) to (\b2+\rightsidedelta-0.17,\yone+0.25);
	\draw [->,thick,-Triangle] (\b3+\rightsidedelta+0.13,\yone+0.25) to (\b3+\rightsidedelta-0.17,\yone+0.25);
	\draw [->,thick,-Triangle] (\b4+\rightsidedelta+0.13,\yone+0.25) to (\b4+\rightsidedelta-0.17,\yone+0.25);
	\draw [->,thick,-Triangle] (\b5+\rightsidedelta+0.13,\yone+0.25) to (\b5+\rightsidedelta-0.17,\yone+0.25);
	\draw [->,thick,-Triangle] (\b6+\rightsidedelta+0.13,\yone+0.25) to (\b6+\rightsidedelta-0.17,\yone+0.25);
	
	\draw [->] (\b1+0.25+\rightsidedelta,\ytwo+0.5) to [out=30,in=-150] (\b8+0.25+\rightsidedelta,\yone-0.02);
	
	\filldraw[color=black, fill=gray!30] (\b1+\rightsidedelta,\ytwo) rectangle ++(0.5,0.5) node [black,midway,font=\bfseries,yshift=0cm] {\footnotesize $9$};
	\filldraw[color=black, fill=gray!30] (\b2+\rightsidedelta,\ytwo) rectangle ++(0.5,0.5) node [black,midway,font=\bfseries,yshift=0cm] {\footnotesize $9$};
	\filldraw[color=black, fill=gray!30] (\b3+\rightsidedelta,\ytwo) rectangle ++(0.5,0.5) node [black,midway,font=\bfseries,yshift=0cm] {\footnotesize $9$};
	\filldraw[color=black, fill=gray!30] (\b4+\rightsidedelta,\ytwo) rectangle ++(0.5,0.5) node [black,midway,font=\bfseries,yshift=0cm] {\footnotesize $9$};
	\filldraw[color=black, fill=gray!30] (\b5+\rightsidedelta,\ytwo) rectangle ++(0.5,0.5) node [black,midway,font=\bfseries,yshift=0cm] {\footnotesize $9$};
	\filldraw[color=black, fill=gray!30] (\b6+\rightsidedelta,\ytwo) rectangle ++(0.5,0.5) node [black,midway,font=\bfseries,yshift=0cm] {\footnotesize $9$};
	\filldraw[color=black, fill=gray!30] (\b7+\rightsidedelta,\ytwo) rectangle ++(0.5,0.5) node [black,midway,font=\bfseries,yshift=0cm] {\footnotesize $9$};
	\filldraw[color=black, fill=gray!30] (\b8+\rightsidedelta,\ytwo) rectangle ++(0.5,0.5) node [black,midway,font=\bfseries,yshift=0cm] {\footnotesize $9$};
	\filldraw[color=black, fill=gray!0] (\b9+\rightsidedelta,\ytwo) rectangle ++(0.5,0.5) node [black,midway,font=\bfseries,yshift=0cm] {\footnotesize $9$};
	
	\draw [->,thick,-Triangle] (\b2+\rightsidedelta+0.13,\ytwo+0.25) to (\b2+\rightsidedelta-0.17,\ytwo+0.25);
	\draw [->,thick,-Triangle] (\b3+\rightsidedelta+0.13,\ytwo+0.25) to (\b3+\rightsidedelta-0.17,\ytwo+0.25);
	\draw [->,thick,-Triangle] (\b4+\rightsidedelta+0.13,\ytwo+0.25) to (\b4+\rightsidedelta-0.17,\ytwo+0.25);
	\draw [->,thick,-Triangle] (\b5+\rightsidedelta+0.13,\ytwo+0.25) to (\b5+\rightsidedelta-0.17,\ytwo+0.25);
	\draw [->,thick,-Triangle] (\b6+\rightsidedelta+0.13,\ytwo+0.25) to (\b6+\rightsidedelta-0.17,\ytwo+0.25);
	\draw [->,thick,-Triangle] (\b7+\rightsidedelta+0.13,\ytwo+0.25) to (\b7+\rightsidedelta-0.17,\ytwo+0.25);
	\draw [->,thick,-Triangle] (\b8+\rightsidedelta+0.13,\ytwo+0.25) to (\b8+\rightsidedelta-0.17,\ytwo+0.25);
	\draw [->,thick,-Triangle] (\b9+\rightsidedelta+0.13,\ytwo+0.25) to (\b9+\rightsidedelta-0.17,\ytwo+0.25);
	
	\end{tikzpicture}
\end{figure}
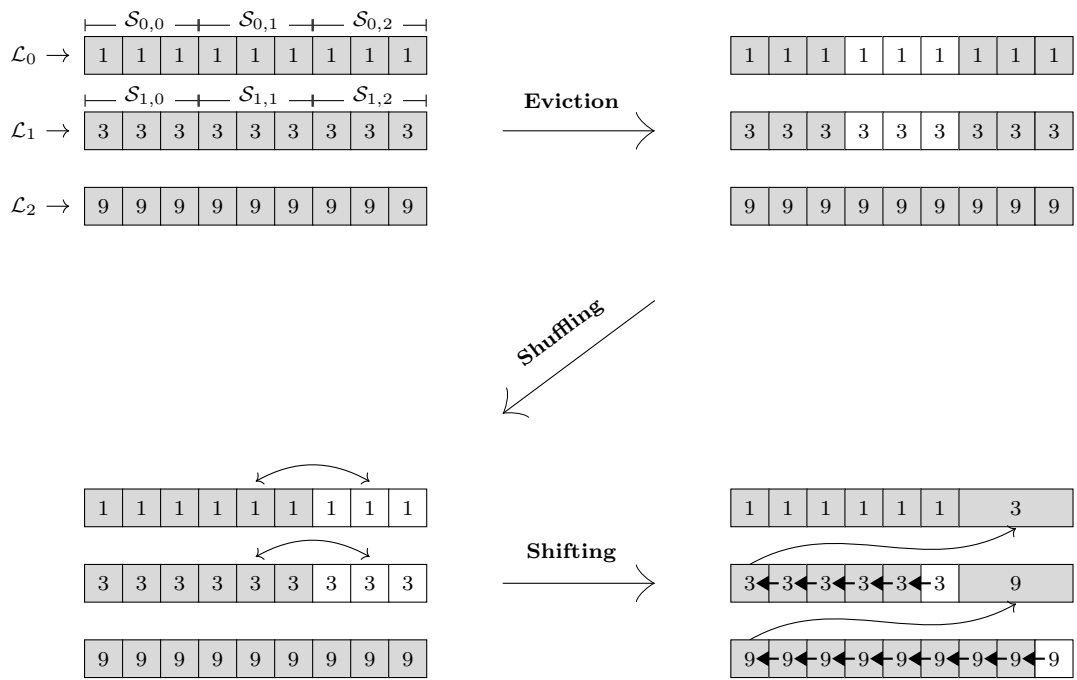

\begin{pstruct}

	W.l.o.g assume that $B=\{1,2,\ldots, m\}$ and $W_1\geq W_2\geq \ldots \geq W_m$. 
	Furthermore, assume the sets $(S_b)_{b\in B}$ are disjoint. 
	We modify $(S_b)_{b\in B}$ using a sequence of steps, eventually obtaining a feasible assignment $(\tilde{S}_b)_{b\in 
		\tilde{B}}$  for $(I,w,\tilde{B}, \tilde{W}, f)$. Figure \ref{fig:shifting} gives an illustration of these steps. 
	
	Define $\tilde{B}_{k+1} = B\setminus \tilde{B}$. We note that  $\tilde{B}_{k+1}$ may be empty. We partition $\{\tilde{B}_j|~ 0\leq j \leq k+1\}$ into {\em levels} and {\em super-bloks}. We consider each $N^2$ consecutive blocks to be a {\em level}, and each $N$ consecutive blocks within a level to be a {\em super-block}. Formally, level $t$ is
	$$\mL_t =\left\{j~| ~t\cdot N^2 \leq j <\min\{(t+1)N^2, ~k+2 \}\right\}$$ 
	for $0\leq t \leq \ell$ with $\ell = \floor{\frac{k+1}{N^2}}$. 	We note that if $\ell=0$, then $(S_b)_{b\in B}$ is a feasible assignment for the instance $(I,w,\tilde{B},\tilde{W},f)$ and the claim hold. Thus, we may assume that $\ell\geq 1$.
	The super-block 
	$r$ of level $t$ is  
	$$\mS_{t,r} = \left\{j~|~ t\cdot N^2 +r\cdot N\leq j  < t\cdot N^2 +(r+1)\cdot N \right\}$$
	for $0\leq r<N$ and level $0\leq t <\ell$ (we do not partition the last level into super-blocks). 
	It  follows that $B= \bigcup_{t=0}^{\ell} \bigcup_{j\in \mL_t} \tilde{B}_j$ and $\mL_t=\bigcup_{r=0}^{N -1 } \mS_{t,r}$ for $0\leq t<\ell$. Furthermore, for any $j\in \mL_t$, $j\neq k+1$  it holds that $|\tilde{B}_j| =N^t$ and $|\tilde{B}_{k+1}| < N^{\ell}$. Essentially, all the blocks of level $t$ are of the same size, and the number of bins in a super-block in level $t-1$ is the number of bins in a single block of level $t$. 
	 We use this property in the shifting process where the assignments of items to bins in the first blocks of level $t$ are shifted to the last super-block of level $t-1$. The {\em eviction} and  {\em shuffle} steps, described below, are used to ensure that all bins in the last super-block of each level are empty when shifting is applied.
	%
	%
	
 We modify the assignment  $(S_b)_{b\in B}$ using the following steps. 
	
	\noindent	{\bf Eviction: }
	We first evict a super-block of bins from each level (except the last one). Let $R= \bigcup_{j\in \mL_{\ell}} \bigcup_{b\in \tilde{B}_j} S_b$ be the subset of items assigned to the last level, and let $g:2^I \rightarrow \mathbb{R}_{\geq 0}$ defined by $g(Q) = f(Q \cup R)$. Note that 
	 $g$ is submodular and non-negative (see Claim \ref{lem:submodular_cup}).
	 Also, let $V_{t,r} = \bigcup_{j\in \mS_{t,r }} \bigcup_{b\in \tilde{B}_j } S_b$ be the set of items assigned to super-block $r$ of level $t$ for $0\leq r <N$ and $0\leq t <\ell$. Then,  
	by definition, we have 
	$$f\left( \bigcup_{b\in B} S_b\right) = g\left( \bigcup_{r=0}^{N-1} \bigcup_{t=0}^{\ell-1} V_{t,r}\right) .$$ 
	Furthermore, we note that the sets $\left( \bigcup_{t=0}^{\ell -1} V_{t,r}\right)_{r=0}^{N-1}$ are disjoint. Hence, 
	by Lemma \ref{lem:simple_evict}  there is $0\leq r^*<N$ such that 
	\begin{equation*}
	\begin{aligned}
	g&\left(  
	\bigcup_{
		~0\leq r < N, ~r\neq r^*
	} 
\bigcup_{t=0}^{\ell-1} V_{t,r} \right) 
	\geq 
	\left(
	1-\frac{1}{N}
	\right) \cdot
	g\left( \bigcup_{r=0}^{N-1} \bigcup_{t=0}^{\ell-1} V_{t,r}\right)=
	\left( 1-\frac{1}{N}\right) f\left(\bigcup_{b\in B} S_b\right).
	\end{aligned}
	\end{equation*}

	We define a new 	assignment $(T_b)_{b\in 	B}$  by
	$$ T_b =\begin{cases}
	\emptyset & b\in B_j \textnormal{ for } j\in \mS_{t,r^*} \textnormal { and } 0\leq t < \ell  \\
	S_b~~~ & \text{otherwise}
	\end{cases}
	$$
	for any  $b\in B$. 
	 Thus,
	\begin{equation}
	\label{eq:tilde_S_lb}
	f\left(\bigcup_{b\in B} T_b\right) =
	g\left(  
	\bigcup_{~0\leq r < N-1,~r\neq r^*~} \bigcup_{j\in \mS_{t,r}} 
	\bigcup_{t=0}^{\ell -1} V_{t,r}\right)
	\geq  \left( 1-\frac{1}{N}\right) f\left(\bigcup_{b\in B} S_b\right).
	\end{equation}
	It also holds that $T$ is a feasible assignment for the instance $(I,w,B,W,f)$, since  $T_b\in \{ S_b, \emptyset\}$ for any $b\in B$. By the same argument, it follows that 
	 $\bigcup_{b\in B} T_b\subseteq \bigcup_{b\in B} S_b$.

	\noindent {\bf Shuffling:} We now generate a new assignment $\tilde{T}$ such that $\bigcup_{b\in B} \tilde{T}_b = \bigcup_{b\in B} T_b$, and the last super-block in each level (except the last one) is empty. 
	This property is obtained by moving the assignments of
	the  bins in  super-block $N-1$ to the bins of super-block $r^*$ for every $0\leq t <\ell$. 	
	
	In case $r^* = N-1$ we simply define $(\tilde{T}_b)_{b\in B}$ by $\tilde{T}_b= T_b$ for any $b\in B$. Otherwise we have $r^* \neq N-1$.
	For any $0\leq t<\ell$, 
	let 
	$\varphi_t: \bigcup_{j\in \mS_{t,r^*}} \tilde{B}_j \rightarrow 
		\bigcup_{j\in \mS_{t,N-1}} \tilde{B}_j
	$ be a bijection from  
	the bins of super-block $r^*$ to the bins of the last super-block  of level $t$ (note that both sets have the same cardinality, thus such a bijection exists). 
	We define $\left(\tilde{T}_b\right)_{b\in B}$  by
	$$\tilde{T}_b = \begin{cases}
	\emptyset & \textnormal{$b\in 
		\tilde{B}_j$ for $j\in \mS_{t,N-1}$, $0\leq t<\ell$}\\
	T_{\varphi_t(b)}&
	\textnormal{$b\in \tilde{B}_j$ for $j\in \mS_{t,r^*}$, $0\leq t<\ell$}\\
	T_{b}& \textnormal{otherwise}
	\end{cases}
	$$
	for any $b\in B$. 
	For $b\in \tilde{B}_j$ with $j\in \mS_{t, r^*}$, $0\leq t<\ell$  it holds that $$w\left(\tilde{T}_b\right) = w\left(T_{\varphi_t(b)}\right) \leq W_{\varphi_t(b)} \leq W_b$$
	where the last inequality follows from $\varphi_t(b)> b$.  
	 Also, for any other bin $b\in B$ it holds that $\tilde{T}_b\in \{\emptyset, T_b\}$ 
thus $w(\tilde{T}_b)\leq W_b$.

	In both cases it holds that $\bigcup_{b\in B} \tilde{T}_b = \bigcup_{b\in B} T_b$
	 and  $\tilde{T}$ is a feasible assignment for $(I,w,B,W,f)$.

	\noindent{\bf Shifting:}
		In this step we generate a feasible assignment $(\tilde{S}_b)_{b\in \tilde{B}}$ for the instance $(I,w,\tilde{B}, \tilde{W},f)$. 
	As the bins of the last super-block in each level (except the last level)  are vacant in $\left(\tilde{T}\right)_{b\in  B}$, we use them for the assignment of the first block of the next level.
	This can be done 
	since $N$ blocks of level $t$ contain the same number of bins as a single block of level $t+1$. We also use blocks in levels greater than $0$ which are not in the last  super-block to store the assignment of the subsequent block in the same level. 
	
	Formally, define $(\tilde{S}_b)_{b\in \tilde{B}}$ by
	$$
	\tilde{S}_b =
	\begin{cases}
	\tilde{T}_{b+N^{t+1}} & \textnormal{$b\in \tilde{B}_j$ with $j\in \mS_{t,N-1}$, $0\leq t<\ell$} \\
	\tilde{T}_{b+N^{t}} & \textnormal{$b\in \tilde{B}_j$ with $j\in \mL_{t}\setminus  \mS_{t,N-1}$, $0<t< \ell$}
	\\
	\tilde{T}_{b+N^{\ell}} & \textnormal{$b\in \tilde{B}_j$ with $j\in \mL_{\ell}$, $b+N^{\ell}\leq m$} \\
	\emptyset & \textnormal{$b\in \tilde{B}_j$ with $j\in \mL_{\ell}$, $b+N^{\ell}> m$}\\
	\tilde{T}_b & b\leq N^2-N
	\end{cases} 
	$$
	for any $b\in \tilde{B}$. The first case defines the shifting of the  assignments of the first block of level $t+1$ to the last super-block of level $t$. The second and third cases define the shifting of assignments of  a block to the previous block. The forth case handles the last block in $\tilde{B}$, and the last case indicates that the assignments of the first $N^2-N$ bins remain in place.
	
	Let $b\in \tilde{B}$. If $\tilde{S}_b=\emptyset$ then $w(\tilde{S}_b)=0\leq \tilde{W}_b$. 
	If $b\leq N^2-N$ then $\{b\}= \tilde{B}_{b-1}$. Hence, $w(\tilde{S}_b)=w(\tilde{T}_b) \leq W_b = \tilde{W}_b$.
	 In any other case, there are $0< j <j' \leq k+1$ and $b'\in \tilde{B}_{j'}$ such that $b\in \tilde{B}_j$ and $\tilde{S}_b = \tilde{T}_{b'}$.
	It follows that $W_{b'}\leq \tilde{W}_ b$ by the definition of $\tilde{W}_b$.  Thus, 
	$w(\tilde{S}_b)=w(\tilde{T}_{b'})\leq W_{b'}\leq \tilde{W}_b$. That is,  $(\tilde{S}_b)_{b\in \tilde{B}}$ is feasible for $(I,w,\tilde{B}, \tilde{W},f)$.

	Clearly, $\bigcup_{b\in \tilde{B}} \tilde{S}_b\subseteq \bigcup_{b\in B} \tilde{T}_b$.
	 Let $i\in \bigcup_{b\in B} \tilde{T}_b$; thus, there is $b\in B$ such that $i\in T_b$. 
	 There is $0\leq j \leq k+1$ such that $b\in \tilde{B}_j$. Also, 
	  since $T_b\neq \emptyset$ it holds that  $b\not\in \mS_{t,N -1}$ for all $0\leq t\leq \ell$. If $j\in \mL_0$ then $\tilde{S}_b = \tilde{T}_b$ as $b\not\in
	  \mS_{0,N -1}$, hence $i\in \bigcup_{b\in \tilde{B}} \tilde{S}_b$. Otherwise,  either $j\in \mL_t\setminus S_{t,N-1}$ for $1\leq t< \ell$ or  $j\in \mL_t$ for $t=\ell$ , it can be verified that in both cases $\tilde{S}_{b-N^t}=\tilde{T}_t$; thus, $i\in \bigcup_{b\in \tilde{B}} \tilde{S}_b$. Therefore, it also holds that $\bigcup_{b\in \tilde{B}} \tilde{S}_b = \bigcup_{b\in B} \tilde{T}_b= \bigcup_{b\in B} T_b\subseteq \bigcup_{b\in B} S_b$, and by \eqref{eq:tilde_S_lb}
	  $$f\left(\bigcup_{b\in \tilde{B}} \tilde{S}_b \right)= f\left(\bigcup_{b\in B} \tilde{T}_b \right)\geq \left(1-\frac{1}{N}\right) \cdot f\left(\bigcup_{b\in B} S_b \right).$$ 
	\qed
\end{pstruct}

\subsection{Solving a Continuous Relaxation and Rounding} \label{sec:rounding}

In this section we give Algorithm~\ref{alg:rounding}
which outputs a solution
satisfying Lemma \ref{lem:rounding}.
The input for the algorithm is a $\delta$-restricted SMKP instance along with a partition $B\setminus B^r=  B_1 \cup \ldots \cup B_k$ of the bins, where $B_j$ is a block for all $1\leq j \leq k$.
The algorithm utilizes the {\em block-constraint instance} defined in Section \ref{sec:block_constraint}. We give the algorithm in 
Section~\ref{sec:delta_restricted}.

\subsubsection{The Block-Constraint Instance}
\label{sec:block_constraint}

Recall that a $\delta$-restricted SMKP instance is defined by an SMKP instance  $\II=(I,w,B,W,f)$ and a set of restricted bins $B^r\subseteq B$. Given such an instance,   a partition $ B_1, \ldots ,B_k$ of $B\setminus B^r$ to blocks and $\aeps>0$, we define their associated  {\em block constraint instance} as a triplet $(E,P,g)$, where $E$ is a set, $P\subseteq [0,1]^{E}$ is a polytope and $g:2^{E}\rightarrow \mathbb{R}_{\geq 0}$ is a monotone non-negative submodular function. The instance $(E,P,g)$ defines the optimization problem $\max_{T\subseteq E:~\bx^T\in P} g(T)$; however, this point of view is only used indirectly.\footnote{For a set $T\subseteq E$, we use $\bx^T$ to denote the vector $\bx^T\in \{0,1\}^{E}$ defined by $\bx^T_e = 1$ for $e\in T$, and $\bx^T_e=0$ for $e\in E\setminus T$.} In the following we give the formal definition of the block constraint instance $(E,P,g)$.


For simplicity, let $\{ B_{k+1}, \ldots ,B_{\ell}\} = \{\{b\}|~ b\in B^r\}$ be a set of blocks, each consisting of a single bin. Thus, $B =\cup_{j=1}^{\ell} B_j$.
Denote the (uniform) capacity of the bins in block $B_j$ by $W^*_j$, for $1 \leq j \leq \ell$. That is, for any $b\in B_j$ it holds that $W^*_j = W_b$.
For $1\leq j \leq k$, we
say that an item $i\in I$ is {\em $j$-small} if $w_i \leq \aeps \cdot W_j^*$, otherwise $i$ is {\em $j$-large}. Let $I_j=\{ \{i\}| \textnormal{ $i$ is $j$-small} \}$ for $1\leq j \leq k$. For $k<j\leq \ell$ define $I_j =\{ \{i \}~|~w_i\leq \delta \cdot W^*_j\}$.

A {\em $j$-configuration} is a subset of $j$-large items which can be packed into a single bin in $B_j$. That is, $C\subseteq I$ is a $j$-configuration if every item $i\in C$ is $j$-large and
$w(C) \leq W^*_j$. Let $\cC_j$ be the set of all $j$-configurations for $1\leq j \leq k$ and $\cC_j=\emptyset $ for $k<j\leq \ell$. As any $j$-configuration has at most ${\aeps}^{-1}$ items, it follows that $|\cC_j| \leq |I|^{\aeps^{-1}}$, i.e., the number of configurations is polynomial in the size of ${\II}$.
Furthermore, for 
$A\subseteq I$ such that $w(A)\leq W^*_j$, $1\leq j \leq k$, there are $C\in \cC_j$ and $S\subseteq I$ such that all the items in $S$ are $j$-small and $A= C \cup S$. Our algorithm exploits this property.


The set $E$ is defined by 
$E = \{(S,j)|~ S\in \cC_j \cup I_j,~1\leq j\leq \ell\}$. Informally, the element $(S,j)\in E$ represents an assignment of all the items in $S$ to a single bin $b\in B_j$. The function $g:2^E\rightarrow \mathbb{R}_{\geq 0}$  is defined by
$g(T) = f\left( \bigcup_{(S,j)\in T} S\right)$. By Claim \ref{submodular-extension}, $g$ is a submodular, monotone and non-negative function.

We define the polytope $P$ as follows.
\begin{equation}
\label{eq:polydef}
P =\left\{\bx\in [0,1]^E
\middle|\begin{array}{lcl}
\displaystyle
\sum_{C\in \cC_j} \bx_{(C, j)} \leq |B_j| & & \forall 1\leq j \leq k \\
\displaystyle
\sum_{S\in \cC_j\cup I_j} w(S) \cdot \bx_{(S, j)} \leq |B_j|\cdot W^*_j& & \forall 1\leq j \leq \ell
\end{array}
\right\}\end{equation}

The polytope represents a relaxed version of the capacity constraints over the bins.
For each block $B_j$, $1\leq j \leq k$, we only require  that the {\em total weight} of items assigned to bins in $B_j$ does not exceed the  total capacity of the bins in this block. We also require that
the number of $j$-configurations
selected for $B_j$ is no greater than the number of bins in this block.



Our algorithm for solving $\delta$-restricted SMKP  uses the unified greedy algorithm of~\cite{feldman2011unified} to find $\bx\in P$ such that $G(\bx)$ is of high value, where $G$ is the multilinear extension of $g$. Subsequently, a random set $T$ is  sampled based on $\bx$. 
The set $T$ is then converted to a solution for the original instance 
using Algorithm~\ref{alg:convert_blockconst_SMKP}. The approximation guarantee of the above process relies on the following connection between  the $\delta$-restricted SMKP instance and the block constraint instance. 


\begin{lemma}
	\label{lem:block_constraints_prop}
	Let $\II=(I,w,B,W,f)$ and $B^r$  be an instance of $\delta$-restricted SMKP, $B_1,\ldots, B_k$ a partition of $B\setminus B^r$ to blocks and $\aeps>0$.  Furthermore, let $(E,P,g)$ be the block constraint instance of the above.  Then the following hold:
	\begin{enumerate}
		\item There is $T\subseteq E$, $\bx^T\in P$ such that $g(T) \geq \OPT(\II)$, where $\OPT(\II)$ is the optimal solution value for $\delta$-restricted SMKP instance $\II$ and $B^r$.
		\label{block_const_prop1}
		\item Given $T\subseteq E$ such that
		$\bx^T\in (1-\aeps)\cdot P$, Algorithm~\ref{alg:convert_blockconst_SMKP} returns in polynomial time a feasible solution $(A_b)_{b\in B}$ for
		$\delta$-restricted SMKP instance $\II$ and $B^r$ satisfying $f(\cup_{b\in B} A_b) = g(T)$.\footnote{Given a polytope $Q$ and $\eta \geq 0$ we use the notation $\eta \cdot Q = \{ \eta \bx ~|~ \bx \in Q\}$.}
			
		\label{block_const_prop2}
	\end{enumerate}
\end{lemma}

\begin{algorithm}[H]
	\SetAlgoLined
	\SetKwInOut{Input}{Input}
	\DontPrintSemicolon
	\Input{ A $\delta$-restricted SMKP instance $\II=(I,w,B,W,f)$ and $B^r$,  a partition $B_1,\ldots, B_k$ of $B\setminus B^r$ to blocks  and $T \subseteq E$.}
	
	 Let $\{ B_{k+1}, \ldots ,B_{\ell}\} = \{\{b\}|~ b\in B^r\}$. \;
	Set $A_b=\emptyset $ for every $b\in B$.\;
	Sort the elements $(S,j)$ in $T$ in decreasing order by the $w(S)$ values.\;
	
	\For{ each  $(S,j)\in T$ in the sorted order
	}{
		Set $A_b \leftarrow A_b\cup S$ where $b=\argmin_{b\in B_j} w(A_b)$.
		\label{line:const}\;
	}
	Return $(A_b)_{b\in B}$.
	\label{line:feasible_blockconst}\;
	
	\caption{Employ a Block-Constraint Solution for SMKP}
	\label{alg:convert_blockconst_SMKP}
\end{algorithm}

\begin{proof}
	
	We start by proving part~\ref{block_const_prop1}.
		Let $(A^*_b)_{b\in B}$ be an optimal solution for the $\delta$-restricted SMKP instance, and let $L_j$ be the set of all $j$-large items for $1\leq j \leq k$, and $L_j=\emptyset$ for $k<j\leq \ell$ ( recall we use $\{ B_{k+1}, \ldots ,B_{\ell}\} = \{\{b\}|~ b\in B^r\}$).  Define
		$$T = \left(\bigcup_{j=1}^{k}
		\left\{ (A^*_b \cap L_j, j) ~|~ b\in B_j\right\}
		\right)
		~\cup~
		\left(
		\bigcup_{j=1}^{\ell} \bigcup_{b\in B_j} \left\{ (\{i\}, j)~|~i\in A^*_b\setminus L_j\right\}\right).
		$$
		
		It can be easily shown that $g(T)=f(\cup_{b\in B} A^*_b)$. Furthermore, as $(A^*_b)_{b\in B}$ is a feasible solution, it holds that $\bx^T\in P$.

	 We now prove part~\ref{block_const_prop2}.
	Let $(A_b)_{b\in B}$ be the output of Algorithm~\ref{alg:convert_blockconst_SMKP} for the given input. We first note that $\cup_{b\in B} A_b = \cup_{(S,j)\in T} S$, and thus $g(T) = f(\cup_{b\in B} A_b)$.
	
	For any $b\in B^r$, there is $k<j\leq \ell$  such that $B_j=\{b\}$. Therefore $A_b = \{i| (\{i\}, j)\in T\}$,
	and since $\bx^T\in (1-\aeps)P$, it follows that $w(A_b)\leq W^*_j = W_b$.

	Let $1\leq j \leq k$ and $b\in B_j$. Assume by negation that $w(A_b) > W_b=W^*_j$. Let $(S,j)\in T$ be the last element in $T$ such that $S\neq \emptyset$ and $S$ was added to $A_b$ in Line \ref{line:const}.  We conclude that $w(A_b\setminus S)>0$, as otherwise  $w(A_b)=w(S)\leq W_b$, by the definition of $E$.
	Therefore, there are at least $|B_j|$ elements $(S',j)\in T$ such that $w(S')\geq w(S)$ (else, in the iteration of $(S,j)$ there must be $b\in B_j$ with $A_b=\emptyset$). If $S\in \cC_j$ then $w(S)> \aeps \cdot W^*_j$, and thus
	$$\left|\left\{S'\neq \emptyset|~ (S',j)\in T,~ S'\in \cC_j \right\}\right| ~\geq~
	\left|\left\{S'|~ (S',j)\in T,~ w(S')\geq w(S) \right\}\right|~> ~|B_j|,$$
	contradicting $\bx^T\in (1-\aeps) P$.
	
	Therefore $S\notin \cC_j$, and we conclude that $S=\{i\}$ with $w_i \leq \aeps \cdot W^*_j$. Thus, $w(A_b\setminus S)> (1-\aeps) \cdot W^*_j $.
	Here, $S$ was allocated to  $A_b$ (which is a set of minimum weight). Then,
	for any $b'\in B_j$, we have
	$w(A_{b'})\geq w(A_b)> (1-\aeps)\cdot W^*_j$.
	Thus,
	$$\sum_{(S',j) \in T} w(S')\geq \sum_{b'\in B_j } w(A_{b'}) > |B_j| (1-\aeps)\cdot W^*_j,$$
	contradicting $\bx^T\in (1-\aeps)P$. We conclude that $w(A_b)\leq W_b$.
	
	Also, by definition, we have that for any $b\in B^r$ and $i\in A_b$  it holds that $w_i \leq \delta\cdot W_b$. Hence, $(A_b)_{b\in B}$ is a solution for the restricted SMKP instance.
	\qed
\end{proof}

\subsubsection{An Algorithm for $\delta$-restricted SMKP}
\label{sec:delta_restricted}
We are now ready to present our algorithm for $\delta$-restricted SMKP.
We note that in Line~\ref{line:rand} of Algorithm~\ref{alg:rounding} we use sampling by a solution vector $\bx^*$, as defined in Section~\ref{sec:prelim}.

\begin{algorithm}[H]
	\SetAlgoLined
	\SetKwInOut{Input}{Input}
	\Input{A $\delta$-restricted SMKP instance ${\II}$ and $B^r$, a partition $B_1,\ldots, B_k$ of $B\setminus B^r$ to blocks, and a parameter $\aeps>0$. }
	
	\DontPrintSemicolon
	Let $(E,P,g)$ be the block-constraint instance of $\II$, $B^r$,  $(B_j)_{j=1}^k$ and $\aeps$.
	\;
	Let $G:[0,1]^E\rightarrow \mathbb{R}_{\geq 0}$ be the multilinear extension of $g$.
	Find a solution $\bar{y}^*$ for $\max_{\bx\in  P } G(\bx)$ using the unified greedy of \cite{feldman2011unified}.
	\label{line:cont_greedy}
	\;
	Let $\bx^*= \frac{1-\aeps}{1+\aeps}\cdot\bar{y}^*$ and 
	sample a random set $T\sim\bx^*$. 
	\label{line:rand}
	\;
	
	\uIf{ $T\in (1-\aeps)P$}{
		Use Algorithm~\ref{alg:convert_blockconst_SMKP} to convert $T$ into a solution $(A_b)_{b\in B}$ for $\delta$-restricted SMKP instance  ${\II}$ and $B^r$.  Return  $(A_b)_{b\in B}$.
		\label{line:convert}
		\;
	}
	\Else{
		Return $(A_b)_{b\in B}$ with $A_b=\emptyset$ for every $b\in B$.
		\label{line:empty}
	}
	\label{line:feasible}

	\caption{Solve and  Round}
	\label{alg:rounding}
\end{algorithm}

For the analysis, consider first the running time. We note that, for any $\bar{\lambda}\in \mathbb{R}^{E}$, a vector $\bx \in P$ which maximizes $\bx \cdot \bar{\lambda}$ can be found in polynomial time. Therefore, the continuous greedy in Line \ref{line:cont_greedy} runs in polynomial time. Thus, Algorithm~\ref{alg:rounding} has a polynomial running time.

It remains to show that the algorithm returns a solution of expected value as stated in Lemma~\ref{lem:rounding}.
Similarly to \cite{chekuri2009dependent}, we use submodular concentration bounds within the proof. 
We note it is possible to prove a variant of this lemma using an approach of \cite{kulik2013approximations}.
While eliminating the dependenc on $\upsilon$, this will result in a more involved proof (recall that $\upsilon$ is defined in~(\ref{eq:def_upsilon})).

\newproof{pround}{Proof of Lemma \ref{lem:rounding}}
\begin{pround}
For any $e\in E$ define $X_e$ to be a random variable such that $X_e =1$ if $e\in T$ and $X_e=0$ otherwise. It follows that $(X_e)_{e\in E}$ are independent Bernoulli random variables,  $\E[X_e]= \bx^*_e$ and $T=\{e\in E| X_e=1\}$.

We first consider blocks $k<j\leq \ell$.
Let  $k<j\leq \ell$  and $B_j=\{b\}$.
Since $\bx^*\in \frac{1-\aeps}{1+\aeps}P$, it follows that
$\E\left[ \sum_{(S,j)\in E} {w(S)} \cdot X_{(S,j)}\right]
\leq \frac{1-\aeps}{1+\aeps} \cdot W_b.$
Also, $w_{(S,j)}\leq {\delta\cdot W_b}$ for every $(S,j)\in E$. Using Chernoff's bound (Theorem 3.1 in \cite{gandhi2006dependent}, see also Lemma~\ref{lem:chernoff}), we have
\begin{equation}
\label{eq:restricted_bound}
\begin{aligned}
\Pr\left(\sum_{(S,j) \in T} w(S) > (1-\aeps)W_b \right)
&\leq \exp\left(-\frac{\aeps^2}{3} \cdot \frac{1-\aeps}{1+\aeps} \cdot \frac{1}{\delta} \right) \leq
\exp\left(-\frac{\aeps^2}{12} \cdot \frac{1}{\delta} \right),
\end{aligned}
\end{equation}
where the last inequality follows from  $\aeps\in (0,0.1)$.

Now, let $1\leq j\leq k$.  For every $(S,j)\in E$ it holds that ${w(S)}\leq W^*_j$. Also, since $\bx^*\in \frac{1-\aeps}{1+\aeps}P$, it holds that 
$
\E\left[ \sum_{(S,j)\in E} w(S) \cdot X_{(S,j)}\right]  
\leq \frac{1-\aeps}{1+\aeps}\cdot|B_j| \cdot { W^*_j}
$,
and
$\E\left[ \sum_{(S,j)\in E:~S\in \cC_j} 1 \cdot X_{(S,j)}\right]
\leq \frac{1-\aeps}{1+\aeps}\cdot|B_j|$.
Therefore, by Chernoff's bound (Theorem 3.1 in \cite{gandhi2006dependent} and Lemma~\ref{lem:chernoff}), we have
\begin{eqnarray}
\label{eq:block_weight_bound}
\Pr\left(\sum_{(S,j) \in T} w(S) > (1-\aeps)|B_j| W^*_j \right)
\leq \exp\left(-\frac{\aeps^2}{3} \cdot \frac{1-\aeps}{1+\aeps} \cdot |B_j|\right)
\leq \exp\left(-\frac{\aeps^2}{12}   \cdot |B_j|\right)\\
\label{eq:block_cong_bound}
\Pr\left(\sum_{(S,j) \in T:~S\in \cC_j} 1 > (1-\aeps)|B_j| \right)
\leq \exp\left(-\frac{\aeps^2}{3}\cdot \frac{1-\aeps}{1+\aeps} \cdot |B_j|\right)\leq
\exp\left(-\frac{\aeps^2}{12}\cdot  |B_j|\right).
\end{eqnarray}

By Lemma~\ref{lem:block_constraints_prop}, $\max_{\bz\in P \cap \{0,1\}^E} G(\bz) \geq \OPT(\II)$.
As the unified greedy of \cite{feldman2011unified} yields a $(1-e^{-1}-o(1))$-approximation for the problem of maximizing the multilinear extension subject to a polytope constraint, it follows that
$G(\bar{y}^*) \geq (1-e^{-1 })(1-\aeps)  \OPT(\II)$ (under the assumption that the number of items is sufficiently large). Since the second derivatives of $G$ are non-positive (see \cite{calinescu2011maximizing}) it follows that
\begin{equation}
	\label{eq:G_lb}
	G(\bx^*) \geq\frac{1-\aeps}{1+\aeps} G(\bar{y}^*) \geq  (1-e^{-1 }) \frac{(1-\aeps)^2}{1+\aeps} \OPT(\II).
\end{equation}


For any $(S,j)\in E$ we have $|S|\leq \aeps^{-1}$, and by the submodularity of $f$, $g(\{(S,j)\})- g(\emptyset) \leq \aeps^{-1}  \upsilon$.
Therefore, by the concentration bound  of \cite{chekuri2009dependent} (see Lemma \ref{lem:extension_cher}), we have
\begin{equation}
\label{eq:exp_bound}
\begin{aligned}
&\Pr\left( g(T) \leq (1-e^{-1}) \frac{(1-\aeps)^3}{1+\aeps} \OPT(\II)
\right)
\leq
\Pr\left( g( \{ e\in E| X_e= 1 \}) \leq (1-\aeps)G(\bx^*)
\right)\\
\leq& \exp\left( - \frac{\aeps^{3}  \cdot G(\bx^*)}{2 \upsilon }\right)
\leq  \exp\left( -\frac{\mu^3 (1-e^{-1}) }{2
	\upsilon } \frac{(1-\aeps)^2}{1+\aeps} \OPT(\II)\right)
\leq
\exp\left( -\frac{\mu^3\cdot \OPT(\II)}{16 \cdot
	\upsilon} \right)
\end{aligned}
\end{equation}
The first and third inequality are due to~\eqref{eq:G_lb}.

Let $\omega$ be the event $\bx^T\in (1-\aeps) P$  and $g(T) \geq \frac{(1-\aeps)^3}{1+\aeps}(1-e^{-1})\OPT(\II)$.
By applying the union bound over \eqref{eq:restricted_bound}, \eqref{eq:block_weight_bound}, \eqref{eq:block_cong_bound} and \eqref{eq:exp_bound}, we have
\begin{equation*}
\Pr(\omega) \geq
1-\left(
|B^r| \exp\left(-\frac{\aeps^2}{12}
\frac{1}{\delta}\right)  - 2  \sum_{j=1}^{k} \exp\left(-\frac{\aeps^2}{12} |B_j| \right)
-\exp\left(-\frac{\aeps^3 }{16}  \frac{\OPT(\II)}{\upsilon} \right) \right) =1-\gamma.
\end{equation*}
In case the event  $\omega$ occurs, the algorithm executes Line~\ref{line:convert}, and by Lemma \ref{lem:block_constraints_prop}, $f(\cup_{b\in B}A_b)  = f(T)$. Hence,
\begin{equation*}
\E\left[f(\cup_{b\in B} A_b)\right] = \Pr\left(\omega\right) \cdot
E\left[f(\cup_{b\in B} A_b)
\middle|  \omega
\right]  \geq (1-\gamma) \frac{(1-\aeps)^3}{1+\aeps} (1-e^{-1})\OPT(\II).
\end{equation*}

Also, the algorithm either returns an empty solution when Line \ref{line:empty} executes, or Line \ref{line:convert}	 executes. In the latter case the solution is feasible by Lemma \ref{lem:block_constraints_prop}. Therefore, the algorithm always returns a feasible solution.
\qed
\end{pround}

\section{Discussion}

In this paper 
we presented a randomized $(1-e^{-1}-\eps)$-approximation
for the  monotone submodular multiple knapsack problem. Our algorithm relies on three main building blocks. The structuring technique (Section \ref{sec:structuring}) which converts a general instance to a leveled instance, the reduction to the block-constraint instance (Section \ref{sec:block_constraint}) and a refined analysis of known algorithms for submodular optimization with a $d$-dimensional knapsack constraint (Section \ref{sec:delta_restricted}).
While the structuring technique and the refined analysis seem to be fairly robust, the reduction to the block-constraint instance proved to be limiting for some natural generalizations of our problem.

One notable example is the {\em non-monotone} submodular multiple knapsack problem, in which the set function $f$ is non-monotone. Unfortunately, in this case
the function $g$ used for solving the block-constraint problem is not submodular.
In a recent work~\cite{FKS20} we overcome this obstacle by applying the structuring technique and refined analysis in conjunction with a new notion of {\em fractional grouping} to circumvent the block reduction. 

While our algorithm yields an almost optimal $(1-e^{-1}-\eps)$-approximation  for SMKP in polynomial time, the dependence of the running time on $\eps$ renders it purely theoretical.
We note that the same holds also for the algorithms of \cite{sun2020tight}.
The problem of finding a polynomial time $(1-e^{-1})$-approximation for SMKP, i.e.,  eliminating the dependence on $\eps$, remains open.   


It would be interesting to extend our results for SMKP to instances in which the profit and weight of an item may depend on the bin to which it is assigned. More specifically, consider 
the {\em monotone submodular  generalized assignment problem (SGAP)}
defined as follows. The input is 
a set of bins $B$ of capacities $(W_b)_{b\in B}$, a set of items $I$, a weight $w_{i,b}$ for every pair $(i,b)$, where $i \in I$ and $b \in B$ 
(the weight of item $i$ when assigned to bin $b$), and a monotone submodular function $f:2^{I\times B}\rightarrow \mathbb{R}_{\geq 0}$. A feasible solution is a set of pairs of items and bins $A\subseteq I\times B$ such that, for any $b\in B$, 
$\sum_{\{i |(i,b)\in A\}} w_{i,b}\leq W_b$
and, for any $i\in I$, $| \{b~|~(i,b)\in A\}| \leq 1$. That is, the total weight of items assigned to a bin does not exceed its capacity, and an item cannot be assigned to multiple bins. The objective is to find a feasible solution $A$ such that $f(A)$ is maximized. We note that the special case where $f$ is modular is the well known {\em generalized assignment problem (GAP)} (see, e.g.,~\cite{fleischer2011tight, feige1995approximating}). It is likely that an $\alpha$-approximation for SGAP can be derived using ideas of~\cite{vondrak2008optimal} for $\alpha$ strictly smaller than $1-e^{-1}$. We conjecture that SGAP admits a $(1-e^{-1})$-approximation.

\bibliography{bibfile}
\appendix
\section{Basic Properties of Submodular Functions}
\label{app:submodular}

\newcommand{\bbeta}{{\bar{\beta}}}


\begin{claim}
	\label{lem:submodular_cup}
	Let $f:2^I \rightarrow \mathbb{R}$ be a set function and $R\subseteq I$. Define $g:2^I\rightarrow \mathbb{R}$ by $g(S)=f(S\cup R)$ for any $S\subseteq I$. Then,
	\begin{enumerate}
		\item If $f$ is submodular then $g$ is submodular.
		\item If $f$ is monotone then $g$ is monotone.
	\end{enumerate}
\end{claim}
\begin{proof}~

	\begin{enumerate}
		\item Assume $f$ is submodular. Let $S, T \subseteq I$. Then,
		\begin{align*}g(S) +g(T) &= f(R\cup S) +f(R\cup T) \\
		&\geq f \left( (R\cup S )\cup (R \cup T)\right) + f \left( (R\cup S )\cap (R \cup T)\right) \\
		&= f \left( R\cup (S \cup T)\right) + f \left( R\cup (S \cap T)\right)\\
		&=g(S\cup T)+ g(S\cap T )
		\end{align*}
		Hence, $g$ is submodular. 
		\item Assume $f$ is monotone, and let $S\subseteq T \subseteq I$. Then $R\cup S \subseteq R\cup T$, and therefore
		$$g(S) =f(R\cup S )\leq f(R\cup T) = g(T).$$
		Thus, $g$ is monotone.
	\end{enumerate}\qed
\end{proof}

%
%
\begin{claim}\label{monotone-submodular}
	Let $f:2^I\rightarrow\mathbb{R}_{\geq 0}$ be  monotone and submodular function, then for any
	$T_1\subseteq T_2\subseteq I$ and $A\subseteq I$, it holds that
	$f(T_1\cup A)-f(T_1)\geq f(T_2\cup A)-f(T_2)$. 
\end{claim}
\begin{proof}
	By the submodularity of $f$, we have
	\begin{equation}
	\label{eq:f_submod}
	f(T_1\cup A) +f(T_2) \geq f(T_1\cup A \cup T_2) + f( (T_1 \cup A) \cap T_2) \geq
	f(T_2\cup A) +f(T_1),
	\end{equation}
	where the second inequality follows from $T_1\subseteq (T_1\cup A) \cap T_2$ and the monotonicity of $f$. By rearranging the terms in (\ref{eq:f_submod}), we have
	$$f(T_1\cup A) - f(T_1) \geq f(T_2 \cup A) - f(T_2)$$
	as required.
	\qed
\end{proof}

\begin{claim}
	\label{submodular-residual}
	Let $f:2^I \rightarrow\mathbb{R}_{\geq 0}$ 
	be a non-negative, monotone and submodular function, and let $S\subseteq I$. Then, $f_S$ is a submodular, monotone and non-negative function.
\end{claim}

\begin{proof}
	By Claim \ref{lem:submodular_cup}, it holds that the function $g:2^I\rightarrow \mathbb{R}$, defined by $g(R)=f(S\cup R)$ for any $R\subseteq I$, is monotone and submodular. Since $f_S$ is the difference between $g$ and a constant,  it is  submodular and monotone as well. 
	
	It remains to show $f_S$ is non-negative.   
	As $f$ is monotone, for any $R\subseteq I$ it holds that $f(S\cup R)\geq f(S)$. Thus, $f_S(R)= f(S\cup R)-f(S)\geq 0$. \qed
	
%
%
%
\end{proof}

\newproof{psub}{Proof of Claim \ref{submodular-extension} }
\begin{psub}
It is easy to see that $g$ is non-negative, as $f$ is non negative. In addition, for any two subsets $A\subseteq B\subseteq E$, we have $ \cup_{(S,h)\in A}S\subseteq \cup_{(S,h)\in B}S$. Thus, since $f$ is monotone, $g$ is monotone as well.

We now show that $g$ is submodular. Consider subsets $A\subseteq B\subseteq E$ and $(S,h)\in E\setminus B$.
\begin{align*}
g(A\cup\{(S,h)\})-g(A) &= f(\cup_{(S',h')\in A} S'\cup S) - f(\cup_{(S',h')\in A}S
') \\
&\leq f(\cup_{(S',h')\in B}S'\cup S) - f(\cup_{(S',h')\in B} S') \\
&=g(B\cup\{(S,h)\})-g(B).
\end{align*}
The inequality follows from Claim \ref{monotone-submodular} and $\cup_{(S',h')\in A} S' \subseteq \cup_{(S',h')\in B} S'$.
\qed
\end{psub}


\label{app:omitted}




\newproof{psimp_evict}{Proof of Lemma \ref{lem:simple_evict}}
\begin{psimp_evict}
	We can write
	\begin{equation*}
	\begin{aligned}
	h(S_1 &\cup \ldots \cup S_N)  -h(\emptyset)= \sum_{r=1}^N
	\left( h(S_1\cup \ldots \cup S_r) -h(S_1 \cup \ldots \cup S_{r-1})\right).
	\end{aligned}
	\end{equation*}
	Therefore, there is $1\leq r^* \leq N$ such that 
	\begin{equation*}
	h(S_1\cup \ldots \cup S_{r^*} )- 
	h(S_1\cup \ldots \cup S_{r^*-1} )\leq \frac{1}{N} \left(h(S_1 \cup \ldots \cup S_N)  -h(\emptyset) \right).
	\end{equation*}
	Thus,
	\begin{equation*}
	\begin{aligned}
	\left( 1- \frac{1}{N} \right)  &\left( h(S_1 \cup \ldots \cup S_N)  -h(\emptyset)\right)
	\\
	&\leq h(S_1 \cup \ldots \cup S_N) - h(\emptyset) - h(S_1 \cup \ldots \cup S_{r^*}) + h(S_1 \cup \ldots \cup S_{r^*-1})  \\
	&=h\left( \left( S_1 \cup \ldots \cup S_N\setminus S_{r^*} \right) \cup \left( S_1 \cup \ldots \cup S_{r^*}\right)\right) \\
	&~~~~+ h\left( \left( S_1 \cup \ldots \cup S_N\setminus S_{r^*} \right) \cap \left( S_1 \cup \ldots \cup S_{r^*}\right)\right)\\
	&~~~~-h(S_1 \cup \ldots \cup S_{r^*}) 
	-h(\emptyset)\\
	&\leq h\left( S_1 \cup \ldots \cup S_N\setminus S_{r^*} \right)
	+h\left( S_1 \cup \ldots S_{r^*}\right)
	-h(S_1 \cup \ldots \cup S_{r^*}) 
	-h(\emptyset)\\
	&=h\left( \bigcup_{1\leq r \leq N, ~r\neq r^*} S_r \right) -h(\emptyset),
	\end{aligned}
	\end{equation*}
	where the second inequality follows from the submodularity of $h$. The first and last equalities  use the property that $S_1,\ldots, S_N$ are disjoint. By rearranging the terms, we have
	$$
	h\left( \bigcup_{1\leq r \leq N, ~r\neq r^*} S_r \right) \geq 
\left( 1- \frac{1}{N} \right)  h(S_1 \cup \ldots \cup S_N)  +\frac{1}{N}\cdot h(\emptyset) \geq  
	\left( 1- \frac{1}{N} \right)  h(S_1 \cup \ldots \cup S_N). 
	$$
	\qed
\end{psimp_evict}

\section{Chernoff's Bound}

In the analysis of Algorithm~\ref{alg:rounding} we use the following Chernoff-like bounds.
\begin{lemma}[Theorem 3.1 in \cite{gandhi2006dependent}]
	\label{lem:chernoff}
	Let $X= \sum_{i=1}^{n} X_i \cdot \lambda_i$ where $(X_i)_{i=1}^{n}$ is a sequence of independent Bernoulli random variable and $\lambda_i \in [0,1]$ for $1\leq i \leq n$. Then
	for any $\eps \in (0,1)$ and $\eta\geq \E[X]$ it holds that
	$$\Pr\left(X > (1+\eps) \eta \right)< \exp\left(- \frac{\eps^2}{3} \eta\right)$$
\end{lemma}
\begin{lemma}[Theorem 1.3 in \cite{chekuri2009dependent}]
	\label{lem:extension_cher}
	Let $I=\{1,\ldots,n\}$, $\upsilon>0$ and  $f:2^I \rightarrow \mathbb{R}_+$ be a monotone submodular function such that $f(\{i\})-f(\emptyset) \leq \upsilon$ for any $i\in I$. Let $X_1,\ldots, X_n$ be independent random variables and $\eta= \E[f(\{i\in I| X_i=1\})]$. Then for  any $\eps >0$ it holds that
	$$\E[f(\{i\in I| X_i =1\}) \leq (1-\eps) \eta] \leq \exp\left(-\frac{\eta \cdot \eps^2}{2\upsilon}\right)$$

\end{lemma}

\end{document}